\documentclass{article}
\input{header/packages-arxiv}

\newtheorem{theorem}{Theorem}[section]
\newtheorem{claim}[theorem]{Claim}
\newtheorem{corollary}[theorem]{Corollary}
\newtheorem{lemma}[theorem]{Lemma}
\newtheorem{definition}[theorem]{Definition}

\newcommand{\E}{\mathbf{E}}

\newcommand{\eps}{\varepsilon}
\newcommand{\prob}[1]{\Pr \left( #1 \right)}
\newcommand{\EE}[1]{\E \left[ #1 \right]}

\newcommand{\bbN}{\mathbb{N}}

\newcommand{\calA}{\mathcal{A}}
\newcommand{\calB}{\mathcal{B}}

\newcommand{\calD}{\mathcal{D}}
\newcommand{\calE}{\mathcal{E}}
\newcommand{\calF}{\mathcal{F}}

\newcommand{\calP}{\mathcal{P}}
\newcommand{\calQ}{\mathcal{Q}}

\newcommand{\calT}{\mathcal{T}}

\newcommand{\calX}{\mathcal{X}}

\newtheorem{observation}{Observation}[section]

\DeclareMathOperator{\depSize}{call-count}

\newcommand{\ALCA}{\calA_{\textbf{LCA}}}

\newcommand{\our}{\textsc{Pruned}\xspace}
\newcommand{\ourPivot}{\textsc{Pruned Pivot}\xspace}

\newcommand{\pivot}{\textsc{Pivot}\xspace}
\newcommand{\Pivot}{\pivot}
\newcommand{\cluster}{\textsc{cluster}\xspace}
\newcommand{\ourCluster}{\textsc{\our-\cluster}\xspace}
\newcommand{\floor}[1]{\lfloor #1 \rfloor}
\newcommand{\ceil}[1]{\lceil #1 \rceil}
\newcommand{\td}{\tilde{d}}
\newcommand{\rb}[1]{\left( #1 \right)}
\newcommand{\IDsofar}{\textsc{ID}_{\text{so-far}}}
\newcommand{\recCalls}{\textsc{rec-calls}}
\newcommand{\opt}{\textsc{OPT}\xspace}

\newlength{\ourAlgIndent}
\newcommand{\ourIndent}[1][1]{%
    \foreach \n in {1,...,#1}{%
        \hspace{\ourAlgIndent}%
    }%
}

\DeclareMathOperator{\poly}{poly}
\newcommand{\knote}[1]{{\color{red}$\ll$Kostya: #1$\gg$}}

\begin{document}

\title{Pruned Pivot: Correlation Clustering Algorithm for Dynamic, Parallel, and Local Computation Models}
\author{
    Mina Dalirrooyfard \\ Machine Learning Research \\ Morgan Stanley 
    \and Konstantin Makarychev \\ Department of Computer Science \\ Northwestern University
    \and Slobodan Mitrovi\' c \\ Department of Computer Science \\ UC Davis}
\date{}

\maketitle
\begin{abstract}
Given a graph with positive and negative edge labels, the correlation clustering problem aims to cluster the nodes so to minimize the total number of between-cluster positive and within-cluster negative edges. This problem has many applications in data mining, particularly in unsupervised learning.
Inspired by the prevalence of large graphs and constantly changing data in modern applications, we study correlation clustering in dynamic, parallel (MPC), and local computation (LCA) settings. 
We design an approach that improves state-of-the-art runtime complexities in all these settings.
In particular, we provide the first fully dynamic algorithm that runs in an expected amortized constant time, without any dependence on the graph size.
Moreover, our algorithm essentially matches the approximation guarantee of the celebrated \textsc{Pivot}
algorithm.
\end{abstract}
\section{Introduction}
We study algorithms for the Correlation Clustering problem, which has many applications in Machine Learning and Data Mining~\cite{bansal2004correlation,becker2005survey,kalashnikov2008web,arasu2009large,firman2013learning,bonchi2013overlapping,li2017motif}.
Among the most prominent applications is clustering products into categories or detecting communities based on product co-purchasing~\cite{wang2013scalable,veldt2020parameterized,shi2021scalable}.
In this problem, we are given a set of objects with ``similar" or ``dissimilar" labels between every pair of objects, and the goal is to cluster these objects such that similar objects are in the same cluster and dissimilar objects are in different clusters. Formally, given a complete graph with edge weights in $\mathbb{R}$, correlation clustering with the \emph{minimum disagreement} objective asks to cluster the nodes such that the sum of the weights of positive edges between clusters plus the sum of the weights of negative edges inside clusters is minimized.\footnote{Correlation clustering has also been studied on weighted graphs and with other objectives such as \emph{maximum agreement} or \emph{minimum $\ell_p$ norm}
\cite{bansal2004correlation,swamy2004correlation,charikar2005clustering,demaine2006correlation,puleo2016correlation,
charikar2017local,ahmadi2019min,jafarov2021local,kalhan2019correlation}.} This paper focuses on the unweighted setting, where weights are in $\{-1,+1\}$. 

Correlation Clustering is APX hard~\cite{charikar2005clustering}. There has been a long line of work on approximation algorithms for correlation clustering; see e.g., \citet{bansal2004correlation,charikar2005clustering,demaine2006correlation,chawla2015near,jafarov2021local,cohen2022correlation,BCMT,chakrabarty2023single,cohen2023handling}.
The best known approximation factor is $1.437$ due to \citet*{cao2024understanding}.\footnote{The proof of $1.437$ approximation is computer-assisted, but the same publication also presents an analytical proof of $1.49$ approximation.} However, all known algorithms with approximation factors less than $3$ use linear programming (LP), which makes most of them impractical for dealing with massive data.  

In their seminal work, \citet*{pivot} introduced an elegant $3$-approximation algorithm called \Pivot, which runs in linear time (time proportional to the number of positive \emph{edges} in the graph). This algorithm is the algorithm of choice in practice. The algorithm has been adapted for various computational models including semi-streaming \cite{behnezhad2023single,cambus2022parallel,chakrabarty2023single}, parallel algorithms (MPC)~\cite{cambus2022parallel,BCMT}, local computation algorithms (LCA) \cite{BCMT}, and dynamic algorithms (\citet{BCMT}; see also \citet{chechik2019fully}).

We study correlation clustering for massive and dynamic graphs. Such graphs are used to represent 
social networks~\cite{tantipathananandh2011finding,hafiene2020influential}, knowledge graphs~\cite{fang2020dynamic,yan2021dynamic}, and user-product interactions~\cite{ding2019user}.
We design a \pivot-like algorithm that can be easily implemented in the fully dynamic regime and LCA and MPC models. Our algorithm is inspired by the recent works by \citet*{BCMT} and 
\citet*{chakrabarty2023single}.
\citet{BCMT} presented a $(3+\eps)$-approximate algorithm, called $R$-pivot, that runs in $O(\frac{1}{\eps})$ MPC rounds, for any $\eps>0$. 
This algorithm can be implemented in LCA with $\Delta^{O(1/\eps)}$-probe complexity, where $\Delta$ is the maximum node degree in the graph (consisting of positive edges). \citet{chakrabarty2023single} give a $(3+\eps)$-approximate semi-streaming algorithm that uses $O(n/\varepsilon)$ words of memory.
\citet{behnezhad2019fully} show how to maintain the lexicographically first maximal matching in a fully dynamic graph. Their result can be used to implement \Pivot in the fully dynamic setting. The expected update time for relabelling edges is $O(\log^2{n}\log^2{\Delta})$ per operation. \citet{chechik2019fully} provide a similar result for maximal matching with expected worst-case running time $O(\log^4{n})$ per update.

\subsection{Our Contributions}
In this paper, we provide a new variant of \Pivot, which we call \ourPivot, that gives a $(3+\eps)$-approximation for Correlation Clustering (see \cref{thm:approx-factor}).
Our algorithm is local and parallelizable by design: Given a node $u$ and common randomness, it returns the cluster of $u$ after exploring only $O(1/\eps)$ nodes of the entire graph. 
This makes it easy to implement \ourPivot in various computational models, including dynamic algorithms, MPC, and LCA.

Our first result is an efficient algorithm for dynamically maintaining a clustering. This is the first dynamic algorithm for Correlation Clustering, whose expected running time does not depend on the graph size. 

\begin{theorem}[Fully-dynamic correlation clustering]\label{cor:dynamic}
For any $\eps>0$, there is a data structure that maintains a $3+\eps$ approximation of correlation clustering in a fully-dynamic setting with an oblivious adversary. The expected update time is $O(1/\eps)$ per operation. 
\end{theorem}
\cref{cor:dynamic} gives an almost $3$ approximation fully dynamic algorithm with update time $O(1)$, which answers an open question posed by \citet{BCMT}.

\begin{theorem}[Correlation clustering in MPC]\label{cor:MPC}
    For any $\eps>0$, there is a randomized $O(\log{\frac{1}{\eps}})$-round MPC algorithm that achieves a $3+\eps$ approximation for Correlation Clustering. This bound holds even when each machine has a memory sublinear in the node-set size.
\end{theorem}
The previously best-known MPC algorithm by \citet{BCMT} requires $O(1/\eps)$ rounds. Hence, our approach improves the dependence on $1/\eps$ exponentially.
\begin{theorem}[Correlation clustering in LCA]\label{cor:LCA}
For any $\eps>0$, there is a randomized $O(\Delta / \eps)$-probe complexity local computation algorithm that achieves a $3+\eps$ approximation of correlation clustering on graphs with maximum 
degree~$\Delta$.
\end{theorem}
\cref{cor:LCA} gives an almost $3$ approximation LCA in, essentially, $O(\Delta)$ probes, thus answering another question posed by \citet*{BCMT}. 
The previously best-known algorithms for LCA were the algorithm by \citet{BCMT} giving $3+\eps$ approximation in $\Delta^{O(1/\eps)}$ probes, and
the work by \citet{behnezhad2023single} providing a $5$-approximate algorithm in $O(\poly \log n)$ space and $O(\Delta \cdot \poly \log n)$ probes.\footnote{Our LCA algorithm uses $O(\Delta / \eps \cdot \poly \log n)$ space.}

We provide empirical evaluations on synthetic graphs in \cref{sec:empirical}. They show that exploring only $4$ nodes to obtain a node's clustering suffices for the cost of \ourPivot to be within $1\%$ of the cost of \pivot.

Moreover, in \cref{section:pram} we describe how to implement our algorithm in the CRCW PRAM model in $O(1/\eps)$ rounds.

We finally note that our main technical result (\cref{thm:bound-X-Q}) is of independent interest.
A seminal work by \citet*{yoshida2009improved} shows how to learn whether a randomly chosen node is in MIS using the average-degree LCA probes.
Our work essentially recovers that claim using a significantly different analysis.

\subsection{Comparison to prior work}
Several closely related works have introduced variants of the \pivot algorithm. In \pivot, nodes are processed according to a predefined and random order, and each node queries its neighbors that have already been processed to determine its cluster (for a formal description of \Pivot see \cref{sec:pivot}). Hence, to determine the cluster of each node, multiple   ``query paths" are made, and the collection of these query paths makes a ``query tree". (For a more formal definition of query paths and query trees, refer to \cref{sec:querypaths}.)

First, given a parameter $R$, the $R$-\pivot algorithm by \citet{BCMT} runs the \pivot algorithm but only considers ``query paths" of depth at most $R$. If, to obtain the cluster of a node, one needs to consider query paths of length more than $R$, then this node is put into a singleton cluster. \citet{BCMT} show that this algorithm has approximation factor $3+O(\frac{1}{R})$.

Second, the semi-streaming algorithm by~\citet{chakrabarty2023single} is another modified version of \pivot. In this approach, given a parameter $R$, every node only queries at most its $R$ top-ranked neighbors when deciding on their cluster. This algorithm yields an approximation factor of $3 + O(\frac{1}{R})$. 

In \ourPivot, instead of only limiting the number of neighbors each node queries or the depth of the query paths, our algorithm limits \emph{the total size} of the query tree. Our analysis of the algorithm uses a new approach to counting query and \emph{expensive} paths, which is very different from the approach of \citet{BCMT}.

Our main technical contribution is a proof that by limiting the total size of the query tree, the existence/non-existence of each edge only influences at most a constant number of other nodes. This crucial point allows us to achieve low probe complexity in LCA and MPC, and constant update time in the dynamic setting. 

\section{Preliminaries}
An instance of the correlation clustering problem receives an unweighted graph $G = (V, E)$ on input. We consider $E$ representing positive and $(V \times V) \setminus E$ representing negative labels between the nodes of $V$.
This problem aims to cluster $V$ to minimize the number of positive between-cluster and negative within-cluster labels.
The neighbors of a node $u \in V$ are denoted by $N(u)$.
We let $u\in N(u)$, i.e., $u$ is a neighbor of itself.
Next, we formally define the MPC and LCA models.

\textbf{The MPC model.} 
Massively Parallel Computation (MPC) is a theoretical model of real-world parallel computation such as MapReduce \cite{dean2008mapreduce}.
It was introduced in a sequence of works by \citet{dean2008mapreduce,karloff2010model,goodrich2011sorting}.
In MPC, computation is performed in synchronous rounds, where in each round every machine locally performs computation on the data that resides locally and then sends and receives messages to any other machine. Each machine has a memory of size $S$, and can send and receive messages of total size $S$. As the local computations frequently run in linear or near-linear time, they are ignored in the analysis of the complexity of the MPC model, and so the efficiency of an algorithm in this model is measured by the number of rounds it takes for the algorithm to terminate where the memory $S$ plays a key role. 
We focus on the \emph{sublinear memory} regime, where $S=n^{\alpha}$ for some constant $\alpha\in (0,1)$.

\textbf{The LCA model.} Local Computation Algorithms (LCAs) were introduced by \citet{rubinfeld2011fast} for tasks where the input and output are too large to be stored in the memory. An LCA is not required to output the entire solution but should answer queries about a part of the output by examining only a small portion of the input.
In Correlation Clustering, the query is a node $v$, and the output is the cluster ID of $v$. Formally, an LCA $A$ is given access to the adjacency list oracle for the input graph $G$, a tape of random bits, and local read-write computation memory. When given an input query $x$, $A$ must compute an answer for $x$ depending only on $x, G$ and the random bits. The answers given by $A$ to all possible queries must be consistent, meaning that they must constitute some valid solution to the computation problem. 

We use \emph{probe} to refer to accessing a node in an adjacency list. The LCA complexity of an algorithm is measured by the number of probes the algorithm makes per single query.

\section{Recursive and Pruned Pivot}\label{sec:pivot}
This section describes our variant of the \pivot algorithm that we call \ourPivot. 
We remind the reader how the standard \pivot algorithm works. First, it picks a random ordering $\pi:V\to\{1,\dots,n\}$. We say that $\pi(u)$ is the rank of node $u$. If $\pi(u)<\pi(v)$, then $u$ has a higher rank than $v$. 
Therefore, the node with rank $1$ is the highest-ranked, and the node with rank $n$ is the lowest-ranked node.
The algorithm maintains a list of not yet clustered nodes. Initially, all nodes are not clustered. At every step, the algorithm picks the highest not yet clustered node, marks it as a pivot, and assigns itself and all its not yet clustered neighbors to a new cluster. The algorithm labels all nodes in this new cluster as \emph{clustered} and proceeds to the next step. Each cluster created by the \pivot algorithm contains a unique pivot node. We say that the cluster is represented by that pivot. If node $u$ belongs to the cluster represented by pivot $v$, we say that $u$ is assigned to pivot $v$. 

To describe our variant of the \pivot algorithm, we first rewrite the standard \pivot as a recursive or top-down dynamic programming algorithm. The algorithm relies on the recursive function \cluster (see \cref{alg:recursive-pivot}). 
For a given node $u$ and random permutation $\pi$, this function returns the pivot node to which $u$ is assigned, along with a flag indicating if $u$ is a pivot. Note that $u$ is a pivot if and only if it is assigned to itself.
\begin{algorithm}[h]
\caption{\textsc{Recursive Pivot} \label{alg:recursive-pivot}}
\begin{algorithmic}[1]
    \STATE \textbf{function} $\cluster(u,\pi)$:
        \STATE \ourIndent Sort all neighbors of $u$ (including $u$ itself) by their rank $\pi(v)$. Denote the sorted list by~$N_{\pi}$.
        \STATE \ourIndent \textbf{for all} $v$ in $N_{\pi}$:
        \STATE \ourIndent[2] \textbf{if} $v = u$:
        \STATE \ourIndent[3] \textbf{return} $u$ belongs to the cluster of $u$; $u$ is a pivot.
        \STATE \ourIndent[2] $\cluster(v,\pi)$
        \STATE \ourIndent[2] \textbf{if} $v$ is a pivot: 
        \STATE \ourIndent[3] \textbf{return} $u$ is in the cluster of $v$; $u$ is not a pivot. 

    \end{algorithmic}
\end{algorithm}


To reduce the running time, we can cache (memoize) the values returned by the function \cluster. We want to use this recursive function in our local computation algorithm (LCA). The problem is, however, that to cluster some nodes, the algorithm may need to make as many as $\Omega(n)$ calls to \cluster (for instance, if node $u$ is connected to all nodes in the left part of the complete bipartite graph $K_{n,n}$).
That is why we propose a crucial change: execute only $k$ recursive calls of \pivot. If the status of the node is not determined by then, mark that node as \emph{unlucky} and make it a singleton. The algorithm is given below.

\begin{algorithm}[h]
\caption{\ourPivot \label{alg:our-pivot}}

\begin{algorithmic}[1]
\STATE Initialize a global variable $\recCalls$ to $0$.

\STATE \textbf{function} \ourCluster($u$,$\pi$):
\STATE \ourIndent \textbf{If} $\recCalls \geq k$:
\STATE \ourIndent[2] terminate this recursion
\STATE \ourIndent Sort all neighbors of $u$ (including $u$ itself) by their rank $\pi(v)$. Denote the sorted list by~$N_{\pi}$.
\STATE \ourIndent \textbf{for all} $v$ in $N_{\pi}$:
\STATE \ourIndent[2] \textbf{if} $v = u$:
\STATE \ourIndent[3] \textbf{return} $u$ belongs to the cluster of $u$; $u$ is a pivot. 
\STATE \ourIndent\ourIndent\;$\recCalls \gets \recCalls  + 1$ 
\STATE \ourIndent[2] \our-\cluster($v$,$\pi$) 
\STATE \ourIndent[2] \textbf{if} $v$ is a pivot: 
\STATE \ourIndent[3] \textbf{return} $u$ is in the cluster of $v$; $u$ is not a pivot. 
    \end{algorithmic}
\end{algorithm}



The recursion tree for the modified \textsc{Pruned-cluster} function 
contains at most $k$ edges. Consequently, if $k$ is a constant, the running time of function \textsc{Pruned-cluster} is also constant. We show how to implement this algorithm as a Local Computation (LCA), Massively Parallel Computation (MPC), and Dynamic Graph Algorithm. In the next section, we prove that the approximation factor of \ourPivot is $3 + O(1/k)$.

\section{Sequential Implementation}\label{sec:seq-imp}
In the previous section, we described \ourPivot algorithm. \emph{For the sake of analysis}, we now examine a sequential algorithm that produces the same clustering as the recursive algorithm above and, moreover, marks the same set of nodes as unlucky. 
First, we consider the standard \pivot implemented as a bottom-up dynamic programming algorithm (see \cref{alg:sequential-pivot}).

\begin{algorithm}[h]
\caption{\textsc{Sequential \Pivot} \label{alg:sequential-pivot}}
\begin{algorithmic}[1]
    \STATE Pick a random ordering $\pi: V \to \{1,\dots,n\}$. 
    \STATE Let $V_{\pi}$ be the list of all nodes $u\in V$ sorted by the rank $\pi(u)$.
    \STATE \textbf{for each} $u$ \textbf{in} $V_{\pi}$:
        \STATE \ourIndent Sort all neighbors of $u$ by their rank $\pi(v)$. Denote the sorted list by $N_{\pi}$.
        \STATE \ourIndent \textbf{while} $u$ is not assigned to a cluster:
        \STATE \ourIndent[2] Pick the next neighbor $v \in N_{\pi}(u)$. 
        \STATE \ourIndent[2] \textbf{if} $v$ is a pivot: place $u$ in the cluster of $v$.
        \STATE \ourIndent[2] \textbf{if} $v = u$: mark $u$ as a pivot; create a new cluster for $u$; and place $u$ in that cluster.
    \end{algorithmic}
\end{algorithm}


In the main loop (see the \textbf{for each} loop above), the algorithm iterates over all nodes in $V$. At iteration $i\in\{1,\dots,n\}$, the algorithm processes node $u$ with rank $i$, i.e., $u=\pi^{-1}(i)$. 
It checks all neighbors $v$ of $u$ with rank higher than that of $u$. If one of these neighbors is a pivot, the algorithm assigns $u$ to the highest-ranked pivot neighbor of $u$. 
If none of these neighbors are pivots, the algorithm marks $u$ as a pivot and assigns $u$ to itself.

Let us set up some notation. Consider a neighbor $v$ of $u$. It is processed at iteration $i=\pi(v)$. Suppose that no other neighbor of $u$ (including $u$ itself) is marked as a pivot before iteration $i$. Then, we know that $u$ will be assigned to the cluster of $v$, since it is the highest-ranked pivot neighbor of $u$. 
Thus, we will say that $u$ is \emph{settled} at step $i$. In other words, $u$ is settled when the first neighbor of $u$ is marked as a pivot. We denote the iteration when $u$ is settled by $\sigma(u)$. Note that node $u$ is assigned to the node processed at iteration $\sigma(u)$, i.e., node $\pi^{-1}(\sigma(u))$. In particular, if $u$ is a pivot, then it is settled at the iteration $i=\pi(u)$, the same iteration as it is processed. We always have $\sigma(u)\leq \pi(u)$, because if $u$ is not settled before iteration $\pi(u)$, then it is marked as a pivot and assigned to itself at iteration $\pi(u)$; thus, if 
$\sigma(u) \geq \pi(u)$, then
$\sigma(u) = \pi(u)$.

If neighbor $v$ of $u$ is considered in the \textbf{while}-loop of the \textsc{Sequential Pivot} algorithm, then we say that $u$ queries $v$. We denote by $Q(u)$ the set of all neighbors queried by $u$, except for $u$ itself, and call this set the set of \emph{queried neighbors of $u$}.  Observe that
$
Q(u) = \{v \in N(u)\setminus\{u\}: \pi(v) \leq \sigma(u)\}.
$
That is, $Q(u)$ is the set of all neighbors of $u$, excluding $u$, whose rank is higher than the rank of the pivot to which $u$ is assigned. Finally, we formally define the recursion tree $\calT_u$  for node $u$. The definition is recursive: If $Q(u)$ is empty, then $\calT_u$ only contains node $u$. Otherwise,  $\calT_u$ is the tree with root $u$  and $|Q(u)|$ subtrees $\widetilde\calT_v$ attached to it -- one tree for every $v\in Q(u)$. Each $\widetilde\calT_v$ is a copy of the recursive tree  $\calT_v$.
We stress that the recursive tree may contain multiple copies of the same node~$v$. One can think of  nodes of $\calT_u$ as being ``stack traces'' or ``execution paths'' for the recursive function \cluster.

\medskip

\noindent\textbf{Sequential Pivot with Pruning.} We now describe how to modify the bottom-up algorithm to make it equivalent to \ourPivot algorithm. First, we run the bottom-up algorithm as is and record its trace. We then define the recursive call count for every node $u$. The recursive call count of $u$ equals the number of edges in the recursive tree $\calT_u$. It can be computed using the following recurrence relation:
\begin{equation}\label{eq:def-dep-size}
\depSize(u) = \sum_{v \in Q(u)}(1 + \depSize(v)).
\end{equation}
If $Q(u)$ is empty, then the recursive call count of $u$ equals $0$, by definition. 
We mark node $u$ as \emph{unlucky} if its recursive call count is at least $k$. Note that if one of the queried neighbors of $u$ is unlucky, then $u$ is also unlucky. 

\begin{algorithm}[h]
\caption{\textsc{Pruning} \label{alg:pruning}}
\begin{algorithmic}[1]
    \STATE Compute the recursive call count of every node $u$ using recurrence relation~(\ref{eq:def-dep-size}).
    \STATE Mark all nodes $u$ with $\depSize(u) \geq k$ as unlucky.
    \STATE Create a new cluster for each unlucky node $u$, remove $u$ from its current cluster, and place $u$ in the new cluster.

    \end{algorithmic}
\end{algorithm}


The pruning step puts all unlucky nodes into singleton clusters. We refer to the standard \pivot algorithm as \textsc{Pivot} without pruning or 
simply \textsc{Pivot}. We refer to the \pivot algorithm that runs the pruning as \textsc{Pivot with Pruning}. Note that the \textsc{Pivot with Pruning} algorithm produces the same clustering output as the \ourPivot algorithm described in the previous section. The key difference between these algorithms lies in their structure:  \textsc{Pivot with Pruning} consists of two distinct steps -- a \textsc{Pivot} step followed by a \textsc{Pruning} step -- whereas \ourPivot combines both steps together. We show that the expected cost of the \textsc{Pivot with Pruning} is $(3+O(1/k)) \opt$ and obtain the following theorem.

\begin{theorem}\label{thm:approx-factor}
The expected cost of the clustering produced by the \textsc{Pruned Pivot} is $(3+O(1/k))\;\opt$.
\end{theorem}

\citet*{pivot} showed that the approximation factor of \pivot is $3$. By Lemma~\ref{lem:pivot_equival}, the \textsc{Pivot} step (see Algorithm~\ref{alg:sequential-pivot}) is equivalent to the \pivot algorithm. Hence, its approximation factor is also $3$. The pruning step of \textsc{Pivot with Pruning} removes some nodes (namely, unlucky nodes) from their original clusters and puts them into singleton clusters. 
This pruning step can increase the number of pairs of nodes $(u,v)$ disagreeing with the clustering. Note, however, that if $u$ and $v$ are dissimilar (i.e., not connected with an edge), then the pruning step will never make them disagree with the clustering if they agreed with the original clustering. Thus, the pruning step can increase the objective function only by  separating pairs of similar nodes $(u,v)\in E$. In such case, we say that the pruning step cuts edge $(u,v)$. Specifically, edge $(u,v)$ is cut by the pruning step of \textsc{Pivot with Pruning} if $u$ and $v$ are in the same cluster after the \textsc{Pivot} step of the algorithm, but are separated by the pruning step, because $u$, $v$, or both $u$ and $v$ are unlucky nodes. We say that an edge $(u,v)\in E$ is cut by \textsc{Pivot} (without pruning), if \textsc{Pivot} places $u$ and $v$ in distinct clusters. In the next sections, we show Lemma~\ref{lem:pruning-cut-edges} that states that the expected number of edges cut by the pruning step of \textsc{Pivot with Pruning} is upper bounded by the expected number of edges cut by 
\textsc{Pivot}
divided by $\ceil{(k-1)/2}/2$.
The ``triangle-based'' analysis of \pivot by \citet*{pivot} shows that \pivot cuts at most $2\opt$ edges in expectation. Thus, the pruning step cuts at most 
$4\opt/\ceil{(k-1)/2}$ edges in expectation. We conclude that the expected cost of \textsc{Pivot with Pruning} is at most $(3+4/\ceil{(k-1)/2})\opt$. 

\begin{lemma}\label{lem:pruning-cut-edges}
The expected number of edges $(u,v)$ cut by the pruning step of \textsc{Pivot with Pruning} is upper bounded by the expected number of edges cut by \pivot divided by~$\ceil{(k-1)/2}/2$.
\end{lemma}

\subsection{Query Paths}\label{sec:querypaths}
Our goal now is to prove Lemma~\ref{lem:pruning-cut-edges}.
\begin{figure}[H]
    \centering
    \tikzset{
every tree node/.style={minimum width=2em,draw,circle, fill=lightgray},
blank/.style={draw=none},
edge from parent/.style={draw,edge from parent path={(\tikzparentnode) -- (\tikzchildnode)}},
level distance=1.5cm,
directed/.style={<-,very thick}
}
\begin{tikzpicture}[scale=0.8]
\Tree
[.\node[fill=black, text=white]{\LARGE{$w$}};
\edge[directed];
[.\node[fill=BrickRed, text=white]
{\LARGE{$v$}};
    [.\; [.\; 
       [.\; ] 
       [.\; ]][.\; 
    [.\; ]
    [.\; ]]]
    \edge[directed];
    [.\node[fill=BrickRed, text=white]{\LARGE{}};
    \edge[directed];
    [.\node[fill=BrickRed, text=white]{\LARGE{$b$}};
    [.\; ] 
    [.\; ] 
    \edge[directed];
    [.\node[fill=BrickRed, text=white]{\LARGE{$a$}};
    [.\; ]
    [.\; ]]
    ]    
    ]
    [.\; [.\; ]
         [.\; ]
         ]
    ]
]
]
\end{tikzpicture}
\caption{This figure shows an \emph{extended query path} in the recursion tree $\calT_v$ for node $v$. The path starts with edge $(a,b)$ goes to the root of the tree, node $v$, and then proceeds to node $w$. The path from $a$ till $v$ is a \emph{query path}. The path from $a$ to $w$ \emph{extends} the path from $a$ to $v$. If edge $(a,b)$ is cut by the pivot step of \pivot but edge $(v,w)$ is not cut, then this path is \emph{expensive}. We call it expensive because if $v$ is unlucky, then $(v,w)$ is cut by the pruning step of \textsc{Pivot with Pruning} and the cost of $(v,w)$ is partially charged to this path.}
\label{fig:eq-path-in-tree}
\end{figure}
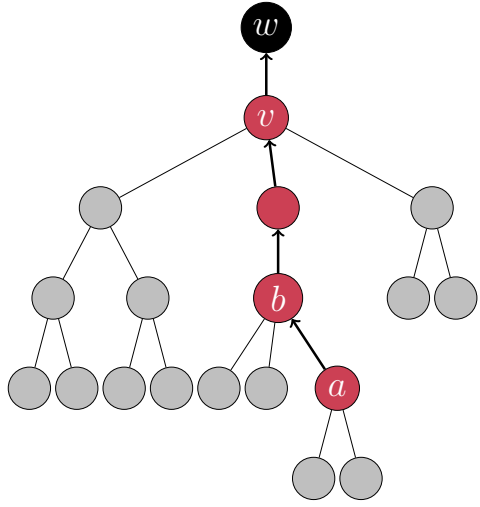

In this section, we define \emph{query paths}, \emph{extended query paths}, and \emph{expensive extended query paths}. We then show that on the one hand, the number of edges cut by the pruning step of the \textsc{Pivot with Pruning} algorithm is upper bounded by the number of expensive extended query paths divided by $\ceil{(k-1)/2}$ (see \cref{cor:lem:EEQPaths-vs-cut-edges}); and, on the other hand, the expected number of expensive extended query paths is upper bounded by two times the expected number of edges cut by the \textsc{Pivot} algorithm (see \cref{thm:bound-X-Q}). This will imply Lemma~\ref{lem:pruning-cut-edges}.

\begin{definition}[Query Paths]\label{def:q-path}
A path $(u_0,u_1,\dots, u_L)$ is a \emph{query path} if each $u_i$ ($i>0$) queries~$u_{i-1}$. 
\end{definition}
\begin{definition}[Extended Query Paths]\label{def:EQ-path}
A path $(u_0,u_1,\dots, u_L)$ of length $L\geq 2$ is an \emph{extended query path} (EQ-path) if the following two conditions hold:
(1) $(u_0,u_1,\dots, u_{L-1})$ is a query path; and  (2) $\pi(u_{L-2}) \leq\sigma(u_L)$.
We say that EQ-path $(u_0,u_1,\dots, u_L)$ is an extension of the query path $(u_0,u_1,\dots, u_{L-1})$.
We also call every path consisting of one edge $(u_0,u_1)$ an extended query path.
\end{definition}
Note that a proper prefix of a query or  extended query path is a query path.

Recall, that for every $u$, we have $\sigma(u)\leq \pi(u)$. Also, a node $u$ queries its neighbor $v$ ($v\neq u)$ if and only if $u$ is not settled before $v$ is processed, i.e., $\sigma(u) \geq \pi(v)$. Thus, $(u_0,u_1,\dots, u_L)$ is a query path if and only if 
\begin{align}
& \sigma(u_0)\leq \pi(u_0)\leq \sigma(u_1) \leq \pi(u_1) \leq\dots \nonumber \\
&\leq\sigma(u_{L-1})\leq\pi(u_{L-1}) \leq\sigma(u_L)\leq\pi(u_L). \label{eq:chain-for-QPath}
\end{align}
Similarly, a path $(u_0,u_1,\dots, u_L)$ of length $L\geq 2$ is an EQ-path if and only if
\begin{align}
& \sigma(u_0)\leq \pi(u_0)\leq \sigma(u_1)\leq \pi(u_1) \leq\dots \nonumber \\
& \leq\sigma(u_{L-2})\leq\pi(u_{L-2})\leq \min(\sigma(u_{L-1}),\sigma(u_L)). \label{eq:chain-for-EQPath}
\end{align}

We will charge all edges cut by the pruning step of \textsc{Pivot with Pruning} to $\Theta(k)$ \emph{expensive EQ-paths} which are defined as follows.
\begin{definition}[Expensive Extended Query Paths]
An extended query path $(u_0,u_1,\dots, u_L)$ is expensive if $\sigma(u_0) < \sigma(u_1)$ but $\sigma(u_{L-1}) = \sigma(u_{L})$. We denote the set of all expensive query paths by $\calX$.
\end{definition}
Note that in every expensive EQ-path, the first edge is cut by \textsc{Pivot} (because $\sigma(u_0) < \sigma(u_1)$) but the last edge is not cut (because $\sigma(u_{L-1}) = \sigma(u_{L})$). 
A path $(u_0,u_1,\dots,u_L)$ is an EQ-path if and only if condition (\ref{eq:chain-for-EQPath}) holds, thus $(u_0,u_1,\dots,u_L)$ is an expensive EQ-path if and only if
\begin{align}
& \sigma(u_0) < \sigma(u_1)\text{ and } \nonumber 
\sigma(u_0)\leq \pi(u_0)\leq \sigma(u_1)\leq \pi(u_1)\leq\nonumber \\
& \dots \leq\pi(u_{L-2})\leq \sigma(u_{L-1}) =\sigma(u_L). \label{eq:chain-for-EEQPath}
\end{align}
The first condition $\sigma(u_0) < \sigma(u_1)$ in~(\ref{eq:chain-for-EEQPath}) can be replaced with 
$\sigma(u_0) \neq \sigma(u_1)$, because we always have $\sigma(u_0) \leq \sigma(u_1)$ if $u_0,\dots,u_{L-1}$ is a query path.

\subsubsection{Charging Cut Edges to Expensive Paths}
We now prove a lemma that establishes a connection between edges cut by the pruning step of \textsc{Pivot with Pruning} and expensive EQ-paths.
\begin{lemma}\label{lem:EEQPaths-vs-cut-edges}
For every unlucky node $v$ and every edge $(v,w)$ with $\sigma(v)=\sigma(w)$, there exist at least $\ceil{(k-1)/2}$ expensive extended query paths that end with edge $(v,w)$.
\end{lemma}
\begin{proof}
Let $\calT_v$ be the recursion tree for node $v$. We first show that $\calT_v$ contains at least 
$\ceil{(k-1)/2}$ edges cut by \textsc{Pivot} (formally, $\calT_v$ contains copies of edges cut by \textsc{Pivot}). Consider an edge $(u',u'')$ in $T$. Since $(u',u'')$ is an edge in the recursion tree, $u'$ queries $u''$. Thus, $u'$ is assigned to $u''$ (if $u''$ is a pivot) or some neighbor of $u'$ which is lower ranked than $u''$ (if $u''$ is not a pivot). Vertex $u''$ is assigned to itself (if $u''$ is a pivot) or one of its neighbors ranked higher than $u''$ (if $u''$ is not a pivot). Thus, if edge $(u',u'')$ is not cut by \textsc{Pivot}, then $u''$ is a pivot and $u'$ is assigned to 
$u''$. This means that $u''$ is the highest-ranked pivot neighbor of $u'$. Consequently, for every $u'$, there is at most one child node $u''$ such that $(u',u'')$ is not cut. Moreover, if one such $u''$ exists, then $u'$ is not a pivot, and hence the edge from $u'$ to its parent is cut (unless $u'$ is the root). We get the following claim.

\begin{claim}\label{cl:at-most-one} For every node $u$ in the recursion tree $\calT_v$, at most one edge incident on $u$ is not cut by \Pivot.
\end{claim}

Node $v$ is unlucky. Hence, the recursion tree $\calT_v$ must have at least $k$ edges. Therefore, by Claim~\ref{cl:at-most-one} and Lemma~\ref{lem:red-blue} (see below) there are at least $\ceil{(k-1)/2}$ cut edges in $\calT_v$. In Lemma~\ref{lem:red-blue}, red edges are cut edges, and blue edges are not cut edges.

\begin{lemma}\label{lem:red-blue}
Consider a tree $T$ with $k$ edges colored red or blue. Suppose that 
at most one blue edge is incident on each node in $T$. Then, $T$ contains at least $\ceil{(k-1)/2}$ red edges.
\end{lemma}
\begin{proof}
Tree $T$ has $k$ edges and $k+1$ nodes. At most one blue edge 
is incident on each node. So, blue edges form a matching. The size of this matching is at most $\floor{(k+1)/2}$. The number of edges not in the matching is at least $k - \floor{(k+1)/2} = \ceil{(k-1)/2}$. All of them are red.
\end{proof}
Now, for every edge $(b,a)$ in $T$ such that $b$ queries $a$ and $(b,a)$ is cut by \Pivot,
we construct an expensive EQ-path. This path starts with edge $(a,b)$, then goes to the root of tree $T$ -- node $v$ -- along the edges of $T$, and, finally, proceeds to node $w$ (see \cref{fig:eq-path-in-tree}). 
Observe that the subpath from $a$ to $v$ is a query path since each node on the path queries the preceding node. We know that $\sigma(v)=\sigma(w)$. 
Hence, by (\ref{eq:chain-for-EQPath}), the path $(a,b,\dots, v,w)$ is an expensive query path.
\end{proof}

The immediate corollary of this lemma gives us a bound on the number of edges cut by the pruning step.
\begin{corollary}\label{cor:lem:EEQPaths-vs-cut-edges}
The number of edges cut by the pruning step of \textsc{Pivot with Pruning} is at most $|\calX|/\ceil{(k-1)/2}$.
\end{corollary}
\begin{proof}
Every edge~$(v,w)\in E$ cut by the pruning step of \textsc{Pivot with Pruning} is not cut by \textsc{Pivot}. Hence, $\sigma(v)=\sigma(w)$. Moreover, if $(v,w)$ is cut by the pruning step, then $v$, $w$, or both $v$ and $w$ must be unlucky. Thus, by Lemma~\ref{lem:EEQPaths-vs-cut-edges}, there are at least $\ceil{(k-1)/2}$ expensive EQ-paths that end with $(v,w)$ or $(w,v)$. Therefore, there exists at least $\ceil{(k-1)/2}$ unique expensive EQ-paths for each edge $(v,w)$ cut by the pruning step.
\end{proof}

\subsubsection{Expected Number of Expensive EQ-Paths}
We now prove that the expected number of expensive EQ-paths is at most $4OPT$ and the expected number of query paths that start with a fixed directed edge $(a,b)$ -- we denote these paths by $\calQ(a,b)$ -- is at most 2.
\begin{theorem}\label{thm:bound-X-Q}
For every ordered pair $(a,b)$ with $(a,b)\in E$, we have 
$\E_{\pi}|\calQ(a,b)| \leq 2$, and
$$\E_{\pi}|\calX| \leq 2\E\Big[\sum_{(u,v)\in E}\mathbf{1}(\sigma(u)\neq\sigma(v))\Big].$$
\end{theorem}
We will refer to the time when iteration $t$ of \textsc{Pivot} occurs as time $t$. For the sake of analysis, we shall assume that the ordering $\pi$ is initially (at time $0$) hidden from us and is revealed one node at a time. At the beginning of iteration $t$, we learn the value of $\pi^{-1}(t)$, or, in other words, the identity of the node processed at time $t$. Note that the state of the algorithm after the first $t$ iterations is completely determined by the nodes $\pi^{-1}(1),\dots,\pi^{-1}(t)$. In particular, at time $t$, for every node $u$, we can tell if $u$ is settled by time $t$ and, if it is settled, then we know the value of $\sigma(u)$; otherwise, we know that $\sigma(u)>t$. Let $\calF_t$ be the filtration generated by 
$\pi^{-1}(1),\dots,\pi^{-1}(t)$. 
We will use the standard notation $\Pr[\;\cdot\mid \calF_t]$ and $\E[\;\cdot\mid \calF_t]$ to denote 
the conditional probability and conditional expectation given the state of the algorithm after iteration~$t$. Note that each $\pi(u)$ and $\sigma(v)$ is a \emph{stopping time} with respect to $\calF_t$.

Let $\calP(a,b)$ be the set of all paths that start with edge $(a,b)$. As we run the \textsc{Pivot} algorithm, we add paths to sets $\calQ_t(a,b)$ and $\calX_t(a,b)$.  Loosely speaking, we add a path from $\calP(a,b)$ to $\calQ_t(a,b)$ if we can verify that this path is a query path using condition (\ref{eq:chain-for-QPath}) at time $t$; we add a path from $\calP(a,b)$ to $\calX_t(a,b)$ if we can verify that this path is an expensive EQ-path using condition (\ref{eq:chain-for-EEQPath}) at time $t$. Formally, we add path $(u_0=a,u_1=b,\dots,u_L)$ to $\calQ_t(a,b)$ at time $\pi(u_{L-1})$ 
if condition~(\ref{eq:chain-for-QPath}) holds; and we add path $(u_0=a,u_1=b,\dots,u_L)$ to $\calX_t(a,b)$ at time $\sigma(u_{L-1})= \sigma(u_{L})$ if condition~(\ref{eq:chain-for-EEQPath}) holds.
Thus, $\calQ_t(a,b)$ is the set of all query paths $P\in\calP(a,b)$ for which 
$\pi(u_{L-1})\leq t$; and $\calX_t(a,b)$ is the set of all expensive EQ-paths $P\in\calP(a,b)$ for which $\sigma(u_{L-1})=\sigma(u_{L})\leq t$.
Note that at the times $\pi(u_{L-1})$ and $\sigma(u_{L-1})=\sigma(u_{L})$, we can check conditions~(\ref{eq:chain-for-QPath}) and (\ref{eq:chain-for-EEQPath}), respectively. We also define a set of \emph{dangerous} paths at time $t$, denoted by $\calD_t(a,b)$, as follows.

\begin{definition}[Dangerous EQ-path]
An extended query path $(u_0,\dots, u_L)$ ($L\geq 1$) is \emph{dangerous} at iteration $t$ if $\pi(u_{L-2})\leq t$, $\pi(u_{L-1})>t$, and $\sigma(u_L) >t$. We omit the first condition ($\pi(u_{L-2})\leq t$) for paths of length $1$. Denote the set of all extended query paths that start with edge $(a,b)$ and are dangerous at iteration $t$ by $\calD_t(a,b)$.
\end{definition}

Note that a path $P\in\calP(a,b)$ may become dangerous at some iteration $t$, stay dangerous for some time, but eventually it will become non-dangerous. After that, it will remain non-dangerous until the end of the algorithm. The definition of dangerous paths is justified by the following lemma, which, loosely speaking, says that every query path and every expensive EQ-path is created from a dangerous path.

\begin{lemma}\label{lem:dangerous-to-QX}
Consider a path $P=(u_0,u_1,\dots,u_L)\in\calP(a,b)$. Let $P'=(u_0,u_1,\dots,u_{L-1})$. Then, the following claims hold for every $t\geq 0$:
\begin{itemize}
\item If $P\in \calQ_{t+1}(a,b)\setminus \calQ_{t}(a,b)$, then $P\in\calD_t(a,b)$ but $P\notin\calD_{t+1}(a,b)$.
\item If $P\in \calD_{t+1}(a,b)\setminus \calD_{t}(a,b)$, then $P'\in\calD_t(a,b)$ but $P'\notin \calD_{t+1}(a,b)$.
\item If $P\in \calX_{t+1}(a,b)\setminus \calX_{t}(a,b)$, then $P\in\calD_t(a,b)$ 
or $P'\in\calD_t(a,b)$ but $P\notin\calD_{t+1}(a,b)$ and $P'\notin\calD_{t+1}(a,b)$.
\end{itemize}
\end{lemma}
We prove this lemma in \cref{app:lem:dangerous-to-QX}.

Our approach to bounding $\E|\calQ_t(a,b)|$ and $\E|\calX_t(a,b)|$ is based on the following idea: At time $t=0$, the path $(a,b)$ is dangerous, and there are no query or expensive EQ-paths that start with $(a,b)$. 
If $P$ is a dangerous EQ-path at time $t$, then at the next iteration, it may be extended to a longer dangerous path, replaced with a query path, and/or created one or more expensive EQ-paths. A dangerous path may also disappear without producing any new dangerous, query, or expensive EQ-paths. For every EQ-path $P$ dangerous at iteration $t$, we will compute the probabilities of creating new paths and derive the desired bounds on $\E|\calQ_t(a,b)|$ and  $\E|\calX_t(a,b)|$.
To make our argument formal, we define two random processes:
\begin{align*}
\Phi_t(a,b) &= 2|\calD_t(a,b)| + |\calQ_t(a,b)|;\\
\Psi_t(a,b) &= 2|\calD_t(a,b)| + |\calX_t(a,b)|.
\end{align*}
We claim that 
$\Phi_t(a,b)$ and $\Psi_t(a,b)$ are supermartingales. That is,
$\E[\Phi_{t+1}(a,b)\mid \calF_t] \leq
\Phi_{t}(a,b)$;
and
$\E[\Psi_{t+1}(a,b)\mid \calF_t] \leq \Psi_{t}(a,b)$.
\begin{lemma}\label{lem:martingale}
Random processes
$\Phi_t(a,b)$ and $\Psi_t(a,b)$ are supermartingales.
\end{lemma}
We prove this lemma in Appendix~\ref{app:lem:martingale}.
We now use it 
to finish the proof of \cref{thm:bound-X-Q}. We first upper-bound $E_{\pi}|\calQ(a, b)|$. Fix a directed edge $(a,b)$. At time~$0$, $\Phi_0(a,b)=2$, since $(a,b)$ is a dangerous EQ-path at time~$0$ but $(a,b)\notin \calQ_0(a,b)$. Process $\Phi_0(a,b)$ is a supermartingale. Hence,
$\E[\Phi_n(a,b)]\leq 2$. Note that at time $n$, there are no dangerous EQ-paths because by time $n$ all nodes are processed and settled. Hence, $$\E[\calQ(a,b)] = \E[\calQ_n(a,b)] = \E[\Phi_n(a,b)]\leq 2.$$

We now upper-bound $\E|\calX|$. Every expensive EQ-path $P=(u_0,\dots,u_L)$ starts with a directed edge $(u_0,u_1)$ and at some iteration $t$ is added to the set $\calQ_t(u_0,u_1)$. We have (in the sum below, we have two terms,
$\E|\calX_n(a,b)|$ and $\E|\calX_n(b,a)|$, for every edge $(a,b)\in E$):
\begin{align*}
\E|\calX|&= \sum_{a,b:(a,b)\in E}
\E|\calX_n(a,b)| 
\\
&= 
\sum_{a,b:(a,b)\in E}\E\big[|\calX_n(a,b)|
\cdot \mathbf{1}(\sigma(a)<\sigma(b))\big].
\end{align*}
Here, we used the definition of expensive PQ-paths: In every expensive path $P$ in $\calP(a,b)$, $\sigma(a)<\sigma(b)$. Thus, if $\sigma(a)\geq \sigma(b)$, then $\calX_n(a,b)=\varnothing$.
Observe that $|\calX_n(a,b)|=\Psi_n(a,b)$ and 
$
\E\big[\Psi_n(a,b)\mid 
\calF_{\sigma(a)}]
\leq
\Psi_{\sigma(a)}(a,b)$, because $\Psi_n$ is a supermartingale. Moreover,
\begin{multline*}
\E\big[\Psi_n(a,b)\cdot \mathbf{1}(\sigma(a)<\sigma(b))\mid 
\calF_{\sigma(a)}]
\leq \\ \leq
\Psi_{\sigma(a)}(a,b)\cdot
\mathbf{1}(\sigma(a)<\sigma(b)),
\end{multline*}
because the event $\{\sigma(a)<\sigma(b)\}$ is in $\calF_{\sigma(a)}$, or, in other words, 
at time $\sigma(a)$, we already know the value of $\mathbf{1}(\sigma(a)<\sigma(b))$, and this value does not change over time. Thus,
\begin{align*}
\E|\calX_n(a,b)| &=
\E\E\big[\Psi_n(a,b)\cdot \mathbf{1}(\sigma(a)<\sigma(b))\mid 
\calF_{\sigma(a)}] \\ &\leq
\E\big[\Psi_{\sigma(a)}(a,b)\cdot
\mathbf{1}(\sigma(a)<\sigma(b))\big].
\end{align*}
It remains to compute $\Psi_{\sigma(a)}(a,b)\cdot
\mathbf{1}(\sigma(a)<\sigma(b))$. If $\sigma(a)<\sigma(b)$, then $a$ is not a pivot (otherwise, we would have $\sigma(a)=\sigma(b)=\pi(a)$). Thus, $\pi(a)>\sigma(a)$, and
the only EQ-path  in $\calP(a,b)$
dangerous at time  $\tau=\sigma(a)$ is the path $(a,b)$. Hence,
$|\calD_{\sigma(a)}(a,b)|=1$. Similarly, there are no expensive PQ-paths in $\calX_{\sigma(a)}(a,b)$, because $\sigma(b)>\tau=\sigma(a)$. Therefore, $\Psi_{\sigma(a)}(a,b)\cdot
\mathbf{1}(\sigma(a)<\sigma(b)) = 2\cdot \mathbf{1}(\sigma(a)<\sigma(b))$, and 
$$
\E|\calX|
\leq 2
\E\Big[
\sum_{(a,b)\in E}
\mathbf{1}(\sigma(a)\neq \sigma(b))\Big].
$$

\subsection{Proof of Lemma~\ref{lem:dangerous-to-QX}}\label{app:lem:dangerous-to-QX}
\begin{proof}[Proof of Lemma~\ref{lem:dangerous-to-QX}]
We remind the reader the definitions of query, extended query, expensive extended query, and dangerous paths in Figures~\ref{fig:query-path}, \ref{fig:eq-path}, \ref{fig:expensive-eq-path}, and \ref{fig:dangerous}.

I. We show that if $P$ is included in the set of query paths at time $t+1$, then $P$ is dangerous at time $t$ but not at time $t+1$. If $P\in \calQ_{t+1}(a,b)\setminus \calQ_{t}(a,b)$, then $P$ is a query path, and $\pi(u_{L-1}) = t+1$. Using condition~(\ref{eq:chain-for-QPath}), we get $\pi(u_{L-2}) < \pi(u_{L-1}) = t+1$; $\sigma(u_L)\geq \pi(u_{L-1}) = t+1$. We have $\pi(u_{L-2}) \leq t$ and $\sigma(u_L)\geq \pi(u_{L-1}) > t$. Also, as every query path, $P$ is an extended query path. Thus, $P$ is a dangerous extended query path at time $t$. However, it is no longer dangerous at time $t+1$, because $\pi(u_{L-1}) = t+1$.

\medskip

\noindent II. Now, we show that if $P$ becomes dangerous at time $t+1$, then its prefix $P'$ is dangerous at time $t$ but not time $t+1$. We first note that the length of $L$ cannot be $1$ because paths $(a,b)$ of length $1$ are dangerous from time $0$ till time $\min(\pi(a),\sigma(b))$ i.e., they are dangerous from the very first iteration of the algorithm and thus cannot become dangerous at some later time $t+1$.

If $P\in \calD_{t+1}(a,b)\setminus \calD_{t}(a,b)$, then the following conditions hold: $P$ is an extended query path;
$\pi(u_{L-2}) = t+1$; $\pi(u_{L-1})> t+1$; and $\sigma(u_L) > t+1$.
To see how we derived $\pi(u_{L-2}) = t+1$, observe that $P\in \calD_{t+1}(a,b)$ implies $\pi(u_{L-2}) \le t+1$, but the fact that $P \notin \calD_{t}(a,b)$, together with $\pi(u_{L-1})> t+1$ and $\sigma(u_L) > t+1$, implies $\pi(u_{L-2}) > t$. 
Since $P'$ is a query path by \cref{def:EQ-path}, we have  $\sigma(u_{L-1})\geq \pi(u_{L-2}) = t+1$ (see~(\ref{eq:chain-for-QPath})). This shows that $P'\in \calD_{t}(a,b)\setminus \calD_{t+1}(a,b)$ if the length of $P'$ is $1$, i.e., if $L=2$.
For $L\geq 3$, we need to additionally check that
$\pi(u_{L-3}) \leq t$. This is the case because $\pi(u_{L-3}) < \pi(u_{L-2}) = t+1$.

\medskip

\noindent III. Finally, we show that if $P$ is included in the set of expensive paths at time $t+1$, then $P$ or its prefix $P'$ is dangerous at time $t$ but neither of them is dangerous at time $t+1$.

If $P\in \calX_{t+1}(a,b)\setminus \calX_{t}(a,b)$, then $P$ is an expensive extended query path and $\sigma(u_{L-1}) = \sigma(u_L) = t+1$. Note that $\pi(u_{L-1})\geq \sigma(u_{L-1})=t+1$. Since $P'$ is a query path, we have $\pi(u_{L-2})\leq \sigma(u_{L-1}) = t+1$. Consider two cases. If $\pi(u_{L-2}) \leq t$, then $P$ is dangerous at time $t$ but not at time $t+1$ (because $\sigma(u_L) = t + 1$). If $\pi(u_{L-2}) = t+1$, then $P'$ is dangerous at time $t$ (because $\pi(u_{L-3})< \pi(u_{L-2}) = t+1$). However, $P'$ is no longer dangerous at time $t+1$ (because $\sigma(u_{L-1}) = t + 1$).
\end{proof}

\subsection{Proof of Lemma~\ref{lem:martingale}}\label{app:lem:martingale}
\begin{proof}[Proof of Lemma~\ref{lem:martingale}]
We first analyze process $\Phi_t(a,b)$. By Lemma~\ref{lem:dangerous-to-QX}, if path $P$ is added to set $\calQ_{t+1}(a,b)$ or 
path $Pw$ becomes dangerous at step $t+1$, then $P$ is dangerous at step $t$. Hence,
\begin{multline*}
\Phi_{t+1}(a,b) - \Phi_{t}(a,b)
= -2|\calD_{t}(a,b)\setminus \calD_{t+1}(a,b)|
 + 2|\calD_{t+1}(a,b)\setminus \calD_{t}(a,b)|
+|\calQ_{t+1}(a,b)\setminus\calQ_t(a,b)|
=\\=
\sum_{P\in\calD_t(a,b)} 
-2
\underbrace{
\cdot \mathbf{1}\{P \notin \calD_{t+1}(a,b)\}
}_{\text{no longer dangerous paths}}
+
\underbrace{
\mathbf{1}\{P \in \calQ_{t+1}(a,b)\}
}_{\text{new query paths}}
 + 
2
\sum_{w\in V}
\underbrace{\mathbf{1}\{Pw \in \calD_{t+1}(a,b)\}
}_{\text{new dangerous paths}}
.
\end{multline*}
We show that the conditional expectation of every term in the sum above given $\calF_t$ is non-positive. Consider a path $P=(u_0,\dots,u_L)\in\calP(a,b)$. 
Let
\begin{align*}
\Delta_{t+1}(P)
&=
-2
\cdot \mathbf{1}\{P \notin \calD_{t+1}(a,b)\}
+
\mathbf{1}\{P \in \calQ_{t+1}(a,b)\}
+ 
2
\sum_{w\in V}
\mathbf{1}\{Pw \in \calD_{t+1}(a,b)\}.
\end{align*}

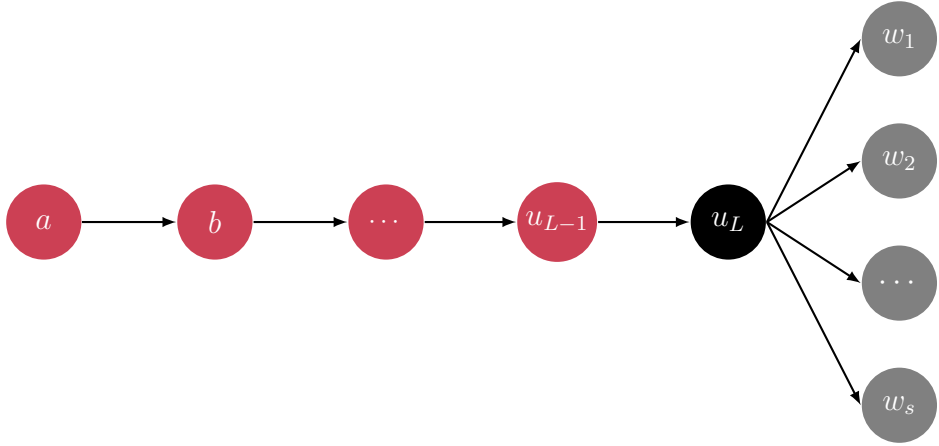
\begin{figure}
    \centering
    \begin{tikzpicture}[scale=0.65,,minimum size=10mm,level distance=35mm,
  sibling distance=25mm,
  every node/.style={fill=BrickRed,text=white, circle,inner sep=2pt},
  arrow/.style={thick,edge from parent/.style={draw,-latex}}
]
\node {\large$a$} child[grow=right,arrow] 
{
node{\large$b$}child[grow=right,arrow] 
{
node{$\cdots$}child[grow=right,arrow] 
{
node{\large$u_{L-1}$}child[grow=right,arrow]{
node[fill=black]{\large$u_{L}$}{
child{node[fill=gray]{\large{$w_s$}}}
child{node[fill=gray]{\large$\cdots$}}
child{node[fill=gray]{\large$w_2$}}
child{node[fill=gray]{\large$w_1$}}
}}}}};
\end{tikzpicture}
    \caption{Illustration for the proof of Theorem~\ref{lem:martingale}. Path $(a,b,\dots,u_{L-1},u_L)$ is a dangerous EQ-path at iteration $t$. At iteration $t+1$, it may become a query path and/or an expensive EQ path. It may also get extended to EQ-paths $Pw$, where $w\in W_t\setminus\{u_{L-1},u_L\}$. These extended paths $Pw$ may be dangerous or expensive at iteration $t+1$, but they also may be non-dangerous and non-expensive at iteration $t+1$.}
    \label{fig:enter-label}
\end{figure}

\begin{figure}
    \centering
    \begin{tikzpicture}[scale=0.65,,minimum size=10mm,level distance=35mm,
  sibling distance=25mm,
  every node/.style={fill=BrickRed,text=white, circle,inner sep=2pt},
  arrow/.style={thick,edge from parent/.style={draw,-latex}}
]
\node {\large$a$} child[grow=right,arrow] 
{
node{\large$b$}child[grow=right,arrow] 
{
node{$\cdots$}child[grow=right,arrow] 
{
node{\large$u_{L-1}$}child[grow=right,arrow]{
node[fill=black]{\large$u_{L}$}{
child{node[fill=gray]{\large{$w_s$}}}
child{node[fill=gray]{\large{$w_{s-1}$}}}
child{node[fill=gray]{\large$\cdots$}}
child{node[fill=gray]{\large$w_2$}}
child{node[fill=gray]{\large$w_1$}}
}}}}};
\draw[thick,color=gray] (11.2,-0.35) --(16.7,-2.5);
\draw[thick,color=gray] (11.1,-0.55) --(16.7,-5);
\end{tikzpicture}
 \caption{Illustration for the proof of Theorem~\ref{lem:martingale}. Path $(a,b,\dots,u_{L-1},u_L)$ is a dangerous EQ-path at iteration $t$. Set $W_t^{(1)}$ contains nodes in $W_t$ that are not neighbors of $u_{L-1}$.
$W_t^{(2)}$ contains nodes in $W_t$ that are neighbors of $u_{L-1}$.
}
    \label{fig:case2}
\end{figure}
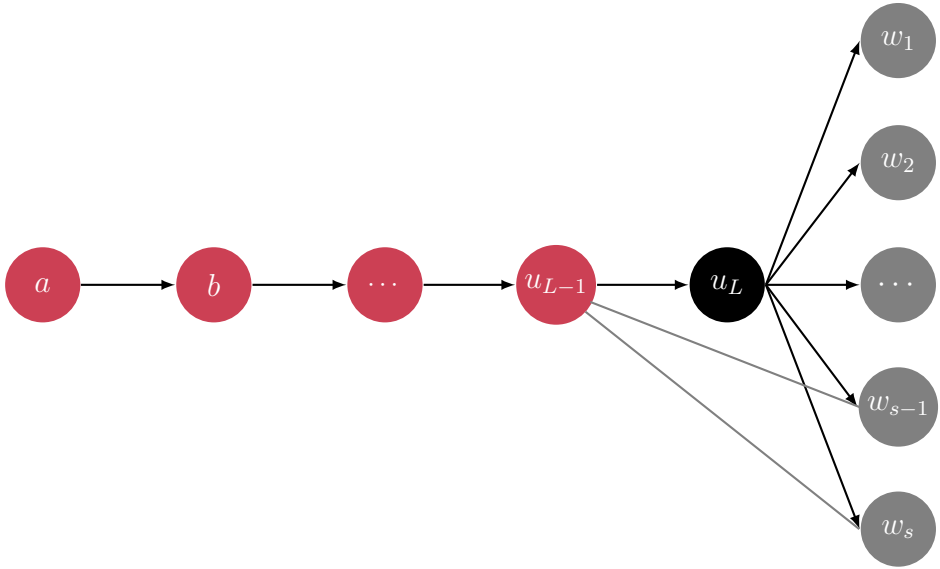

We need to show that
$$
\E\Big[
\Delta_{t+1}(P)
\mid 
\calF_t, P\in\calD_t(a,b)
\Big]\leq 0.
$$
Note that if $P\in \calD_t(a,b)$ and $P\in\calD_{t+1}(a,b)$, then $\Delta_{t+1}(P) = 0$, because 
$\mathbf{1}\{P \in \calQ_{t+1}(a,b)\}
=0$ and $\mathbf{1}\{Pw \in \calD_{t+1}(a,b)\}$ (for all $w$) by 
Lemma~\ref{lem:dangerous-to-QX}; while
$\mathbf{1}\{P \notin \calD_{t+1}(a,b)\}=0$ since $P\in\calD_{t+1}(a,b)$. Thus, it suffices to prove that
$\E\Big[
\Delta_{t+1}(P)
\mid 
\calF_t, 
P\in\calD_t(a,b),
P\notin\calD_{t+1}(a,b)
\Big]\leq 0
$. 
Let $W_t = \big\{w\in N(u_L): \sigma(w) > t\big\}$ be the set of neighbors of $u_L$ that are not settled by iterations $t$. Note that $P$ stops being dangerous at iteration $t+1$ only if  
$u_{L-1}$ is processed or $u_L$ is settled at iteration $t+1$. The latter event occurs if and only if one of the nodes in $W_t$ is processed at iteration $t+1$. Hence, 
\begin{multline*}
\E\Big[
\Delta_{t+1}(P)
\mid 
\calF_t, 
P\in\calD_t(a,b),
P\notin\calD_{t+1}(a,b)
\Big]
=
\E\Big[
\Delta_{t+1}(P)
\mid 
\calF_t, 
P\in\calD_t(a,b),
\pi^{-1}(t+1)\in W_t\cup\{u_{L-1}\}
\Big].
\end{multline*}
Here, $\pi^{-1}(t+1)$ is the node processed at iteration $t+1$. Observe that $\pi^{-1}(t+1)$ is uniformly distributed in $W_t\cup\{u_{L-1}\}$ given that $\pi^{-1}(t+1)\in W_t\cup\{u_{L-1}\}$. Now, if $\pi^{-1}(t+1) =u_{L-1}$, then
$\pi(u_{L-1}) =t+1$, and path $P$ is added to $\calX_{t+1}(a,b)$ (it is an expensive path because 
condition~(\ref{eq:chain-for-QPath}) is satisfied). Moreover, if $\pi^{-1}(t+1) =u_{L-1}$, then
some or all paths $Pw$, where $w\in W_t\setminus \{u_{L-1},u_L\}$, are added to the set of dangerous paths $\calD_{t+1}(a,b)$. Note that no path $Pw'$ with $w'\notin W_t\setminus \{u_{L-1},u_L\}$ is added to 
$\calD_{t+1}(a,b)$ since for 
$w'\in N(u_L)\setminus W_t$ and $w'\notin\{u_{L-1},u_L\}$, we have $\sigma(w')\leq t$, but for $Pw$ in $\calD_{t+1}(a,b)$, we must have $\sigma(w) > t$.
Thus, if $\pi^{-1}(t+1) =u_{L-1}$, then $\Delta_{t+1}(P)\leq -2 +1 + (|W_t|-1)$ (here, we use that $u_L$ always belongs to $W_t$ if $P\in\calD_t(a,b)$). If, however, $\pi^{-1}(t+1) \in W_t\setminus
\{u_{L-1}\}$, then $\Delta_t(a,b)=-2$, because (1) $P$ does not become a query path, and (2) $P$ is not extended to any dangerous paths at iteration $t+1$. We obtain the following bound 
$$\E\Big[
\Delta_{t+1}(P)
\mid 
\calF_t, 
P\in\calD_t(a,b),
\pi^{-1}(t+1)\in W_t\cup\{u_{L-1}\}
\Big]
\leq 
-2 + \frac{1 + 2(|W_t|-1)}{|W_t\cup\{u_{L-1}\}|}<0.
$$
This proves that $\Phi_t(a,b)$ is a supermartingale. We now show that $\Psi_t(a,b)$ is also a supermartingale. Using Lemma~\ref{lem:dangerous-to-QX}, we get
$$\Psi_{t+1}(a,b) -\Psi_t(a,b) = 
\sum_{P\in\calD_t(a,b)} \Delta'(P),
$$
where
\begin{align*}
\Delta'_t(P) 
&=
-2
\underbrace{
\cdot \mathbf{1}\{P \notin \calD_{t+1}(a,b)\}
}_{\text{no longer dangerous paths}}
+2
\underbrace{\sum_{w\in V}
\mathbf{1}\{Pw \in \calD_{t+1}(a,b)\}
}_{\text{new dangerous paths}}
+
\underbrace{
\mathbf{1}\{P \in \calX_{t+1}(a,b)\}
+
\sum_{w\in V}
\mathbf{1}\{Pw \in \calX_{t+1}(a,b)\}
}_{\text{new expensive paths}}
.
\end{align*}
As before, it suffices to show that for every $P\in \calP(a,b)$, we have
$$\E\Big[
\Delta'_{t+1}(P)
\mid 
\calF_t, 
P\in\calD_t(a,b),
P\notin\calD_{t+1}(a,b)
\Big]\leq 0.$$
From the previous argument, we know that $\pi^{-1}(t+1)$ is uniformly distributed in $W_{t}\cup\{u_{L-1}\}$ given $\calF_t$, 
$P\in\calD_t(a,b)$, and
$P\notin\calD_{t+1}(a,b)$. We now consider two cases: $u_{L-1}$ is already settled by iteration $t$
(then $u_{L-1}\notin W_t$) and $u_{L-1}$ is not yet settled (then $u_{L-1}\in W_t$). 

\noindent Case 1: $\sigma(u_{L-1}) \leq t$. In this case,
$P$ may never become an expensive EQ-path, since $\sigma(u_{L-1}) \leq t$, but $\sigma(u_L)> t$. Moreover, if $u_{L-1}$ is processed at iteration $t+1$, then it will not be marked as a pivot, and, consequently, no nodes will be settled at iteration $t+1$. In particular, if $P$ is extended to some EQ-path $Pw$, then this path will not be added to the set of expensive EQ-paths $\calX_{t+1}(a,b)$ (even though it may eventually be added to some set $\calX_{t'}(a,b)$ with $t'>t+1$). However (if $u_{L-1}$ is processed at iteration $t+1$), every path $Pw$ with $w\in W_t\setminus\{u_L\}$  will be added to the set of dangerous EQ-paths $\calD_{t+1}(a,b)$. If another node in $W_t$ -- not $u_{L-1}$ -- is processed at iteration $t+1$, then $P$ is not extended to any other paths. Denote the event 
$\{P\in\calD_t(a,b)\setminus\calD_{t+1}(a,b)\}$
by $\calE_t$. Then,
$$
\E\Big[
\Delta'_{t+1}(P)
\mid 
\calF_t, 
\calE_t,\sigma(u_{L-1}\leq t)
\Big]=
-2 + 2\frac{|W_t|-1}{|W_t|+1} < 0.
$$

\noindent Case 2:  $\sigma(u_{L-1}) > t$. Define two disjoint subsets of $W_t$: 
$W_t^{(1)} = W_t\setminus N(u_{L-1})$ and
$W_t^{(2)} = W_t\cap N(u_{L-1})\setminus\{u_{L-1},u_L\}$. Then,
$W_t$ is the disjoint union of three sets: $W_t^{(1)}$, $W_t^{(2)}$, and 
$\{u_{L-1},u_{L}\}$. 
If $u_{L-1}$ is processed at iteration $t+1$, then (at this iteration) $u_{L-1}$ is marked as a pivot and all nodes in $W_t^{(2)}$ as well as node $u_L$ are settled. Consequently, path $P$ and all paths $Pw$ with $w\in W_t^{(2)}\setminus\{u_L\}$ are added to the set of expensive paths $\calX_{t+1}(a,b)$. No other extensions of $P$ are added to this set. All paths $Pw$ with $w\in W_t^{(1)}$ are added to the set of dangerous paths $\calD_{t+1}(a,b)$. However, paths $Pw$ with $w\in W_t^{(2)}$ are not dangerous by Lemma~\ref{lem:dangerous-to-QX}. Hence, if $u_{L-1}$ is processed at iteration $t+1$, then 
$\Delta'_{t+1}(P) = -2 + 
2|W_t^{(1)}| + (|W_t^{(2)}| + 1)
$. If $u_L$ or any node in $W_t^{(2)}$ is processed at time $t+1$, then path $P$ becomes expensive, but it is not extended to any other EQ-path (since $\sigma(u_L)=t+1 <\pi(u_{L-1})$ and, as a result,
$P$ is not a query path). Hence, 
$\Delta'_{t+1} = -2 + 1$. Finally,
if a node in $W_t^{(1)}$ is processed, then $P$ is not extended to any other paths, and $P$ does not become expensive. Thus, $\Delta'_{t+1} = -2$. Therefore,
\begin{align*}
\E\Big[
\Delta'_{t+1}(P)
\mid 
\calF_t, \calE_t,\sigma(u_{L-1})> t)
\Big]
&=-2 + \frac{2|W^{(1)}_t|+|W^{(2)}_t|+1}{|W_t|} + \frac{|W^{(2)}_t|+1}{|W_t|}\\
&= 
-2 + \frac{2(|W^{(1)}_t|+|W^{(2)}_t|)+2}{|W^{(1)}_t|+|W^{(2)}_t|+2} <0.
\end{align*}
This completes the proof of Lemma~\ref{lem:martingale}.
\end{proof}


\section{MPC Algorithm}

Now we discuss how to simulate \ourPivot in MPC. To that end, observe that \ourPivot has depth at most $k$. 
Moreover, for each node $v$, there are at most $k$ neighbors of $v$, which the computation for $v$ depends on. To find those $k$ neighbors of $v$, we sort all the neighbors of $v$ and select the top $k$ ones. Sorting can be done in $O(1)$ MPC rounds~\cite{goodrich2011sorting} and $O(m)$ total memory for all the nodes simultaneously.

Define a directed graph $H$ on $V$ such that $H$ contains an edge $(u, v)$ if and only if $v$ is among the top $k$ neighbors of $u$. 
The $k$-hop neighborhood of $u$ in $H$ contains all the edges and nodes needed to process $u$ by \ourPivot. Given this, our MPC algorithm simultaneously gathers $k$-hop neighborhood for each $u$. This can be done in $O(\log k)$ MPC rounds and $O(n \cdot k^{k+1})$ total memory via graph exponentiation (see the paper by~\citet{lenzen2010brief}). 
The output for $u$ of \ourPivot is computed on a single machine using its relevant $k$-hop neighborhood.

\section{LCA}\label{sec:LCA}
Given a node $v$, our LCA algorithm simply simulates \ourPivot. The probe complexity of such an approach is almost direct. Namely, the algorithm visits at most $k$ nodes via recursive calls. Each call scans neighbors of the corresponding node to find the $k$-top ones. That scan takes $O(\Delta)$ probes. Therefore, the probe complexity is $O(\Delta \cdot k)$.

It remains to analyze the space complexity and, if any, its effect on the approximation.

\paragraph{Random ordering.}
Our algorithm assumes it has access to a random node-permutation $\pi$. However, it is unclear how to obtain $\pi$ in LCA. So, instead, each node $v$ draws an integer rank $r_v \in [0, n^{10})$ uniformly at random. If values $r_v$ are drawn independently of each other, then, with probability at least $1 - n^{-5}$, for each $u \neq v$ it holds $r_u \neq r_v$. Therefore, the values $r_v$ implicitly define a random permutation, which is enough to simulate our algorithm.
(For now, assume that $r_u \neq r_v$, and at the end of this section, we discuss how to handle the case when $r_u = r_v$ for two nodes.)

\paragraph{Source of randomness.}
If the algorithm has access to an arbitrary long tape of random bits with random access, then whenever a node wants to learn $r_v$, it reads $10 \log n$ bits starting at position $v \cdot 10 \log n$. If such tape is not accessible, then the corresponding local computation algorithm has to store random bits in its memory. It is obvious how to keep all the required random bits in $O(n \log n)$ memory; a node $v$ uses $O(\log n)$ bits independent of other nodes to obtain $r_v$. However, we show that substantially fewer bits suffice for small $\Delta$. To that end, we first recall the definition of $w$-wise independence hash functions and a folklore result about its construction.


\begin{definition}[$w$-wise independent hash functions]
    Let $w, b, N \in \bbN$ and $s$ be a seed of independent random bits. A function $h_s : \{0, 1\}^N \to \{0, 1\}^b$ is a called $w$-wise independent hash function if for any $I \le w$, all distinct $x_1, \ldots, x_I \in \{0, 1\}^N$ and all distinct $y_1, \ldots, y_I \in \{0, 1\}^b$ it holds
    \[
        \prob{\bigwedge_{i = 1}^I h_s(x_i) = y_i} = 2^{-I \cdot b}.
    \]
\end{definition}

\begin{theorem}[Folklore]
Let $w, b, N \in \bbN$. There exists a $w$-wise independent hash function $h_s : \{0, 1\}^N \to \{0, 1\}^b$ with a seed $s$ of length $w \cdot \max\{N, b\}$. Moreover, $h_s$ can be stored using $O(w \cdot (N + b))$ bits of space. 
    
\end{theorem}

Let $h : V \to \{0, 1, 2, \ldots, n^{10} - 1\}$ be a $w$-wise independent hash function; we will set $w$ in the remainder of this section.
Consider an LCA execution of our algorithm from $v$ that, instead of randomly pre-generated $r$-values, uses $h$; let $\ALCA$ refer to that algorithm.
That is, whenever $\ALCA$ needs $\pi(u)$ it uses $h(u)$ instead. So, invoking $\ALCA$ on $v$, the algorithm has to learn the rank-ordering of the neighbors of $v$. 
To achieve that, $\ALCA$ evaluates $h(u)$ on each $u \in N(v)$.

In the process, $\ALCA$ invokes $h$ at most $\Delta \cdot k$ many times.
Hence, if $w \ge \Delta \cdot k$, each invocation of $h$ is entirely independent of the previous invocations of $h$.
Therefore, \emph{from the point of view of $v$}, the execution $\ALCA(v)$ is equivalent to that of executing $\ALCA$ with $r$ values pre-generated using $O(n \log n)$ bits of randomness.

Does $\ALCA$ executed \emph{on all} $v \in V$ resemble our algorithm? Very likely, no! Namely, for $w < n$, we should expect lots of dependencies between execution trees for some, possibly far away, nodes.
Nevertheless, we will argue that, in expectation, the cost of $\ALCA$ and our main algorithm are the same for a proper value of $w$.
To that end, in our main algorithm, let $c_{u, v}$ be the expected cost that pair $\{u, v\}$ incurs.
That is, if $\{u, v\}$ is an edge, then $c_{u,v}$ is the probability that our main algorithm cuts it. If $\{u, v\}$ is not an edge, then $c_{u,v}$ is the probability that our main algorithm does not cut. In particular, the expected cost of our main algorithm is $\sum_{u,v \in V \times V} c_{u, v}$.

Consider $\{u, v\} \in V \times V$.
Then, $\ALCA$ behaves the same with respect to $\{u, v\}$ if $\ALCA(v)$ and $\ALCA(u)$ together resemble the execution of our main algorithm for $v$ and $u$.
This can be achieved by ensuring that the randomness used for $\ALCA(u)$ and $\ALCA(v)$ is independent, which is achieved for $w \ge 2 \cdot \Delta \cdot k$.

In this setup, we also have that in the execution tree for $v$, with probability $1-n^{-5}$ at least, no two ranks are the same. Taking the union bound over all $v \in V$, with probability at least $1 - n^{-4}$, no execution tree has two same ranks.

\paragraph{The case when $r_u = r_v$ for $u \neq v$ in an execution tree.}
As argued above, this case happens with probability at most $n^{-4}$.
We distinguish two cases: when the input graph is a union of cliques and when it is not. In the latter case, the optimum value is at least $1$. Since $r_u = r_v$ happens with probability at most $n^{-4}$ and the maximum possible cost of a clustering is $n^2$, we have that the expected cost in this case is at most $(3 + \eps) OPT + 1/n^2 \le (3 + \eps + 1/n^2) OPT$.

Now, consider the case when the input graph is a union of cliques. In this case, it holds $OPT = 0$, and an additive cost of $1/n^2$ does not yield a $(3+\eps)$ multiplicative one. To handle this case, we add a rule for breaking ties in ranks to our algorithm.
Namely, if a node $x$ has two neighbors $u$ and $v$ such that $r_u = r_v$, then our algorithm ranks $u$ before $v$ iff $u < v$; $u < v$ means that the label of $u$ is smaller than the label of $v$. 
Since each node sees the same neighborhood in a clique -- this neighborhood includes the node itself -- then each node chooses the same pivot even if $r_u = r_v$ for two distinct nodes.
So, each clique is clustered correctly. 
\section{Fully Dynamic Algorithm under Oblivious Adversary}

This section explains how to implement \ourPivot in the fully dynamic setting with an oblivious adversary.
On a high level, when an edge $e$ is updated, our dynamic algorithm simply recomputes the clustering for each node that queries $e$.
As \cref{thm:bound-X-Q} states, there are only $O(1)$ such nodes in expectation. 
This property is the key to enabling us to obtain only $O(k)$ expected amortized update time.

In addition, our algorithm updates the neighbor list of the endpoints of $e$, which is needed to implement \ourPivot.
To obtain the desired running time, we use that \ourPivot visits at most the top $k$ neighbors of a node. 
Therefore, instead of maintaining the entire neighborhood list of a node in a sorted manner, we do it only for the top $k$ neighbors of each node. We show how to dynamically maintain this list in expected amortized $O(\log k)$ time.

We now provide details. Our algorithm maintains the following information for each node~$u$:
\begin{itemize}
    \item $N_k(u)$: Top $k$ neighbors of $u$ kept in an ordered balanced binary tree.
    \item $Q^{-1}_{P}(u)$: The nodes that query $u$, except from $u$ itself, kept in a double-linked list.
    \item $Q_P(u)$: The set of nodes queried by $\ourCluster(\pi, u)$, which is maintained during the recursive calls. 
    Together with every node $w$ queried by $u$, in $Q_P(u)$ is also stored the pointer to where $u$ is in the double-linked list $Q^{-1}_P(w)$. 
    We use these pointers to efficiently update $Q_P^{-1}(\cdot)$.
\end{itemize}

Consider an update of edge $\{a, b\}$, i.e., an edge insertion or removal. This update triggers several updates in the information we maintain for each node.
Without loss of generality, assume that $\pi(a) < \pi(b)$. Note that among all $N_k(\cdot)$, only $N_k(b)$ changes. We show how to update $N_k(b)$ in $O(\log{k})$ expected time in \cref{sec:nk}.

If $N_k(b)$ is changed, but $b$ queries $(b, a)$ neither before nor after this update, then $Q_P$ and $Q^{-1}_P$ structures remain the same as before the update. 
However, if whether $b$ queries $(b, a)$ changes after the update, then all the nodes querying $b$, i.e., those in $Q^{-1}_P(b)$ might change their structures. 
Moreover, only those nodes that query $b$ before the update might change their structures.
To see that, assume that $w$ does not query $b$ before the edge update. First, if $w$ queries $a$, it will not query $(a, b)$ since $\pi(a) < \pi(b)$. Second, and taking into account that $w$ never queries $(a, b)$, if $w$ never reaches $b$ in the invocation of $\ourCluster(w, \pi)$ before the update, then $\{a, b\}$ cannot be queried by $w$ regardless of the update. This yields the following observation.
\begin{observation}
    Consider an edge update $\{a, b\}$ with $\pi(a) < \pi(b)$. Then, only the nodes in $Q^{-1}_P(b) \cup \{b\}$ are those whose maintained data structures can potentially change on the update.
\end{observation}

We remind the reader that $\ourCluster(w, \pi)$ does not use memoization.
In \textsc{\ourPivot-Update$(\cdot,\cdot)$} (see \cref{alg:dynamic}) we provide our update procedure, which uses \cref{alg:update-node} as a subroutine. We analyze its running time in \cref{sec:running-time-dynamic}.

\newcommand{\ANodeUpdate}{\textsc{A-Node-Update}}

\begin{algorithm}[h]
\caption{\ANodeUpdate($w$) \label{alg:update-node}}
\begin{algorithmic}[1]
        \STATE  \textbf{for} every $u \in Q_P(w)$:
        \STATE \ourIndent[1] Remove $w$ from $Q^{-1}_P(u)$ by using the corresponding pointer from $Q_P(w)$.
        \STATE  Invoke \ourCluster($w$,$\pi$) (\cref{alg:our-pivot}) by using $N_k(\cdot)$ as the neighborhood list. During the execution, update $Q_P(w)$.
        \STATE \textbf{for} every $u \in Q_P(w)$:
        \STATE \ourIndent[1] Append $w$ to $Q^{-1}_P(u)$ and record in $Q_P(w)$ the pointer where $w$ is appended to.
    \end{algorithmic}
\end{algorithm}

\begin{algorithm}[h]
\caption{\textsc{\ourPivot-Update$(a,b)$} \label{alg:dynamic}}
\begin{algorithmic}[1]
        \STATE Update $N_k(b)$ as described in \cref{sec:nk}.
        \STATE $\ANodeUpdate(b)$
        \STATE \textbf{if} $b$ queries $(b, a)$: \label{line:if-b-queries}
        \STATE \ourIndent[1] \textbf{for} every $w\in Q^{-1}_P(b)$:
        \STATE \ourIndent[2] $\ANodeUpdate(w)$
    \end{algorithmic}
\end{algorithm}

    


\subsection{Running Time for our Dynamic Algorithm}
\label{sec:running-time-dynamic}

\paragraph{Maintaining $Q^{-1}_{P}(u)$.} 
For each $u \in Q_{P}(w)$, updating double-linked list $Q^{-1}_{P}(u)$ is done in $O(1)$ time -- appending takes $O(1)$ time, while the removal also takes $O(1)$ by using the pointers stored in $Q_{P}(w)$.

Now we upper-bound the expected running time of \textsc{\ourPivot-Update$(a,b)$}.
In that, we use the following claim.
\begin{lemma}\label{lem:S-size}
    The number of nodes in $Q^{-1}_P(b)$ \textbf{processed} by \cref{alg:dynamic} within the \textbf{if} condition is $O(1)$ in expectation. 
\end{lemma}
\begin{proof}
    If $b$ queries $(b, a)$, then each node that queries $b$ also queries $(b, a)$. Hence, the number of nodes in $Q^{-1}_P(b)$ equals the number of nodes querying $(b, a)$. By \cref{thm:bound-X-Q}, that number in expectation is at most $2$.
    Observe that \cref{thm:bound-X-Q} provides an upper-bound for \pivot. Hence, \ourCluster might query $(b, a)$ only less frequently.
\end{proof}

\begin{lemma}
    \textsc{\ourPivot-Update$(a,b)$} takes $O(k)$ amortized time in expectation per an edge update.
\end{lemma}
\begin{proof}
    We show that maintaining $N_k(b)$ takes $O(\log{k})$ amortized time in expectation in \cref{sec:nk}.  

    By \cref{lem:S-size}, in expectation, only $O(1)$ nodes $w$ are processed within the \textbf{if} condition. In addition to them, $b$ is also processed. For each such $w$ or $b$ the following is performed:
    \begin{itemize}
        \item Invocation of \ourCluster. This invocation visits $O(k)$ nodes. To traverse neighbors of a node $u$, \ourCluster uses $N_k(u)$. Even though $N_k(u)$ is organized as a binary balanced tree, all its nodes can be traversed in the rank-decreasing order in $O(1)$ time per node.
        \item Updating $Q^{-1}_P(u)$ for $u \in Q_P(w)$. As described above, this is done in $O(1)$ per $u$. By design of \ourPivot, we have $|Q_P(w)| \le k$.
    \end{itemize}
\end{proof}

\subsection{Maintaining \texorpdfstring{$N_k(u)$}{Nk(u)}}\label{sec:nk}
This section describes how to dynamically maintain $N_k(u)$ with $O(\log k)$ update time in expectation. To maintain $N_k(u)$, as the first step, the algorithm organizes the neighbors of $u$ as we describe next.
At each algorithm step, we associate $\td_u$ with $u$. In particular, $\td_u$ is such that the current degree of $u$ is in the range $(1/4 \cdot \td_u, 4 \cdot \td_u)$. In other words, $\td_u$ is a $4$ approximation of the degree of $u$. 
Moreover, all the neighbors of $u$ are placed in $b_u = \ceil{\td_u / (80 k)}$ buckets numbered $1, \ldots, b_u$. A neighbor $w$ of $u$ is placed in the bucket $j$ such that $(j - 1) \cdot n / b_u < \pi(w) \le j \cdot n / b_u$. 
Hence, a bucket corresponds to $n / b_u$ consecutive integers. This is convenient as if $u$ has $d(u)$ neighbors, then under randomly chosen $\pi$ it holds that in expectation $d(u) / n \cdot n / b_u = d(u) / b_u $ neighbors are in a given bucket. With at least two buckets, this ratio is in the range $(10 k, 320 k)$. When there is a single bucket only, it has at most $320 k$ elements by definition of $b_u$ and $\td_u$.

Each bucket is organized as an ordered balanced binary tree. Therefore, if a bucket contains $t$ elements, then insertion, deletion, and finding the $i$-th-rank node can be done in $O(\log t)$ time.
We use $\calB_u$ to refer to these buckets for node $v$. We note that the number of buckets might change over time. We discuss that towards the end of this section.

\paragraph{Edge insertion.}
If a new edge $\{u, x\}$ is inserted, then the algorithm adds $x$ to the bucket in $\calB_u$ corresponding to $\pi(x)$.
In expectation, that bucket has $\Theta(k)$ nodes. Hence, this operation is done in $O(\log k)$ time in expectation. 
Our algorithm also checks whether $x$ has a higher rank than the lowest rank node in $N_k(u)$. If that is the case, it removes the smallest-rank node from $N_k(u)$ and inserts $x$. This is done in $O(\log k)$ time.

\paragraph{Edge deletion.}
If an edge $\{u, x\}$ is deleted, then the algorithm first removes $x$ from the bucket corresponding to $\pi(x)$. In expectation, that bucket has $\Theta(k)$ nodes. Hence, by the concavity of the $\log$ function, this operation is done in $O(\log k)$ time in expectation.
Second, if $x$ does not belong to $N_k(u)$, then the algorithm does nothing else. Otherwise, the algorithm removes $x$ from $N_k(u)$ and finds the $k$-th highest-ranked element in $\calB_u$. It does so in the following way: it visits bucket by bucket in the decreasing order of ranks until it reaches a bucket containing the desired element. We now analyze the complexity of this search.

Let $Y$ be a random variable representing the cost of this search. Let $Y_B$ be the time spent searching bucket $B$; we count $1$ even if $B$ is accessed but empty. Then, $Y = \sum_{B \in \calB_u} Y_B$. Let $B_j$ denote the $j$-th bucket in $\calB_u$. 
Let $Z_i$ be the event that the buckets $B_1 \ldots B_i$ contain less than $k$ elements in total.
We have
\begin{align*}
 \EE{Y_{B_j}} &= \EE{Y_{B_j}\ |\ Z_{j - 1}} \cdot \prob{Z_{j - 1}} 
 + \EE{Y_{B_j}\ |\ \neg Z_{j - 1}} \cdot \prob{\neg Z_{j - 1}}.
\end{align*}
Observe that $\EE{Y_{B_j}\ |\ \neg Z_{j - 1}} = 0$, as no search is performed on $B_j$ if the buckets $B_1 \ldots B_{j - 1}$ contain at least $k$ elements.
This effectively implies that
\begin{equation}\label{eq:E[Y_B_j]}
    \EE{Y_{B_j}} = \EE{Y_{B_j}\ |\ Z_{j - 1}} \cdot \prob{Z_{j - 1}}.
\end{equation}

Also, we have that 
$$\EE{|B_j| \ |\ Z_{j - 1}} \in O(\td_u / (b_u - (j - 1))).$$
Hence, 
\begin{equation}\label{eq:E[Y_B_j]-given-Z_{j-1}}
    \EE{Y_{B_j}\ |\ Z_{j - 1}} \in O(1 + \log{(\td_u / (b_u - (j - 1)))}).    
\end{equation}

Next, we upper-bound $\prob{Z_{j - 1}}$. Definition of $Z_i$ implies
\begin{align}
\label{eq:expanding-Z_i}
\prob{Z_{i}} &\le \prob{|B_{i}| < k \ |\ Z_{i - 1}} \cdot \prob{Z_{i - 1}} 
\\
&=
\prod_{t = 1}^{i} \prob{|B_t| < k\ |\ Z_{t - 1}}.    
\end{align}
For $i = 1$, we have $\EE{|B_1|\ |\ Z_{0}} \ge d(u) / b_u \ge 10 k$; the latter inequality follows by our discussion above. 
For $i > 1$, we have 
$$\EE{|B_i|\ |\ Z_{i - 1}} \ge \frac{d(u) - k}{b_u - (i - 1)} \ge \EE{|B_1|\ |\ Z_{0}}.$$
The latter inequality can be easily verified algebraically, but also it is an easy observation that it holds as the first $i - 1$ buckets in expectation contain at least $10 k (i - 1)$. Hence, $d(u) - k$ nodes distributed over $b_u - (i - 1)$ buckets yield the bucket-average higher than $10 k$.

Condition on $Z_{i - 1}$. Then, let $X_w$ be a $0/1$ random variable that equals $1$ if and only if the neighbor $w$ of $u$ is in $B_i$; in particular, $\EE{|B_i|\ |\ Z_{i - 1}} = \EE{\sum_{w \in N(u)} X_w\ |\ Z_{i - 1}}$. Observe that the random variables $X$ are negatively correlated. Hence, we can apply Chernoff bound to upper-bound the probability that $|B_i| < k$. Since $\EE{|B_i|\ |\ Z_{i - 1}} \ge 10 k$, we have that 
\[
    \prob{|B_i| < k \ |\ Z_{i - 1}} \le e^{-k}.
\]
Plugging this into \cref{eq:expanding-Z_i}, we derive
\begin{equation}\label{eq:Pr-Z_i}
    \prob{Z_i} \le e^{-ik}.
\end{equation}
We now turn back to computing $\EE{Y}$. By plugging \cref{eq:E[Y_B_j]-given-Z_{j-1}} and \cref{eq:Pr-Z_i} into \cref{eq:E[Y_B_j]}, we obtain
\begin{align}\label{eq:dynamic-Nk-final-expression}
\EE{Y} &= \sum_{j = 1}^{b_u} \EE{Y_{B_j}} 
\\
&\le\notag
\sum_{j = 1}^{b_u} e^{-k (j - 1)}  \cdot O\rb{1 + \log{(\td_u / (b_u - (j - 1)))}}.
\end{align}
To upper bound \cref{eq:dynamic-Nk-final-expression}, we first let $t_u = \td_u / b_u$. By definition, $t_u \in O(k)$. 
Observe that 
$$b_u \le (b_u - (j - 1)) \cdot j$$
when $j$ ranges in the interval $[1, b_u]$; the minimum is achieved for $j = 1$ and $j = b_u$.
This implies that $$\td_u / (b_u - (j - 1)) = t_u \cdot b_u / (b_u - (j - 1)) \le t_u \cdot j.$$ 
From \cref{eq:dynamic-Nk-final-expression}, this yields the upper-bound
\begin{align*}
    \EE{Y} & \le O\rb{\sum_{j = 1}^{b_u} e^{-k (j-1)}  \cdot \rb{1 + \log{\rb{k j}}}} \\
    &= O(1) + O\rb{\sum_{j = 1}^{b_u} e^{-k (j-1)}  \log{k}} \\
    &\quad + O\rb{\sum_{j = 1}^{b_u} e^{-k (j-1)}  \log{j}} \\
\end{align*}
\begin{align*}    
    & \le O(1) + O\rb{\log{k}} + O\rb{\sum_{j = 1}^{b_u} e^{-j}  j} \\
    & \le O(1) + O\rb{\log{k}} + O\rb{1} \\ 
    & = O\rb{\log{k}},
\end{align*}
as desired.
To upper-bound $O\rb{\sum_{j = 1}^{b_u} e^{-j}  j}$ we used that it holds $\sum_{i = 1}^\infty i / 2^i = 2$.

\paragraph{Updating $\calB_u$ when $d_u / \td_u \notin (1/4, 4)$.}
We first describe how to address this change in $O(\log k)$ amortized expected time and then explain how to de-amortize this.

When $d_u$ becomes $\td_u / 4$ or $4 \td_u$, then our algorithm updates the current $\td_u$ to $\td_u'$, and re-creates $\calB_u$ for $\td_u'$. 
If $d_u \le \td_u / 4$, then $\td_u' \gets \td_u / 2$. Similarly, if $d_u \ge 4 \td_u$, then $\td_u' \gets 2 \td_u$.
This approach is standard and typically illustrated through the example of dynamic arrays.
For our problem, this technique yields amortized $O(\log k)$ expected update time.

If memory allocation takes $O(1)$ time, this technique can be de-amortized.
This de-amortization is standard, but we provide a couple of sentences of explanation for the sake of completeness. 
Instead of creating all the buckets from scratch when $d_u = \td_u / 4$ or $d_u = 4 \td_u$, a de-amortized algorithm does that gradually. That is, as soon as it starts updating the buckets for the current $\td_u$ value, it allocates the memory for buckets for both $\td_u' = \td_u / 2$ and $\td_u' = 2 \td_u$; only an allocation is performed, without any initialization.
The algorithm also maintains a variable $\IDsofar$, which, when $\td_u$ is updated, is initialized to $0$.

On a new neighbor $w$ update, if $\pi_w \le \IDsofar$, the algorithm carries over that update for all three bucket structures. If $\pi_w > \IDsofar$, the algorithm only updates the $\td_u$-buckets.
In addition, the algorithm considers $10$ elements from the $\td_u$-buckets with smallest $\pi$ but greater than $\IDsofar$ values, and copies them to the $(\td_u/2)$- and $(2 \td_u)$-buckets. The value of $\IDsofar$ is increased properly so to correspond to the last element copied from $\td_u$-buckets.

\section{CRCW PRAM Algorithm}
\label{section:pram}
We now describe how to execute our approach in CRCW PRAM in $O(1/\eps)$ rounds using $O(n^3)$ processors. We assume that the nodes are numbered $1$ through $n$. The same as discussed in \cref{sec:LCA}, instead of finding a random permutation of the nodes, each node chooses an integer rank from range $[0, n^{10})$. That is, in the shared memory, there is a designated space of $n$ memory blocks, each of size $10 \log n$ bits, such that the $i$-th of those memory blocks stores the rank of node $i$.

Assume that each node has the sorted list of its $k$ highest-ranked neighbors, where $k = O(1/\eps)$ is the parameter used in \cref{alg:our-pivot}. Then, \ourPivot can be directly executed in $O(k) = O(1/\eps)$ rounds in the CRCW PRAM model. It remains to discuss how to find $k$ highest-ranked neighbors of each node in $O(k)$ rounds.

\paragraph{Finding top-$k$ neighbors in $O(k)$ rounds.}
We assume that the neighbors of each node are given as an array. 
\newcommand{\pramMax}{\textsc{ArrayMax}\xspace}
It is well-known how to find the maximum of an array in $O(1)$ CRCW PRAM rounds using $O(n^2)$ processors~\cite{jaja1992parallel,akl1989design}; we use $\pramMax$ to refer to that algorithm.
To find $k$ highest-ranked neighbors of each node, we apply \pramMax $k$ times to each neighborhood array for all the nodes in parallel. The $i$-th of these invocations finds the $i$-th highest-ranked neighbors of the nodes. After the $i$-th highest-ranked neighbor $w_i$ of a node $v$ is found, $w_i$ is marked with a special symbol in the neighborhood array of $v$, implying that in the future invocations of \pramMax, $w_i$ is considered smaller than every other non-marked neighbor of $v$.

Since $\pramMax$ uses $O(n^2)$ processors and is executed for each of the nodes separately, this PRAM implementation of \cref{alg:our-pivot} uses $O(n^3)$ processors.

\section*{Acknowledgements}
We are grateful to anonymous reviewers for their valuable feedback and for suggesting to discuss the implementation of our algorithm in PRAM.
K.Makarychev was supported by the NSF Awards CCF-1955351 and EECS-2216970.
S.~Mitrovi\' c was supported by the Google Research Scholar and NSF Faculty Early Career Development Program \#2340048.
\bibliography{references} 

\begin{thebibliography}{42}
\providecommand{\natexlab}[1]{#1}
\providecommand{\url}[1]{\texttt{#1}}
\expandafter\ifx\csname urlstyle\endcsname\relax
  \providecommand{\doi}[1]{doi: #1}\else
  \providecommand{\doi}{doi: \begingroup \urlstyle{rm}\Url}\fi

\bibitem[Ahmadi et~al.(2019)Ahmadi, Khuller, and Saha]{ahmadi2019min}
Ahmadi, S., Khuller, S., and Saha, B.
\newblock Min-max correlation clustering via multicut.
\newblock In \emph{International Conference on Integer Programming and Combinatorial Optimization}, pp.\  13--26. Springer, 2019.

\bibitem[Ailon et~al.(2008)Ailon, Charikar, and Newman]{pivot}
Ailon, N., Charikar, M., and Newman, A.
\newblock Aggregating inconsistent information: ranking and clustering.
\newblock \emph{Journal of the ACM (JACM)}, 55\penalty0 (5):\penalty0 1--27, 2008.

\bibitem[Akl(1989)]{akl1989design}
Akl, S.~G.
\newblock \emph{The design and analysis of parallel algorithms}.
\newblock Prentice-Hall, Inc., 1989.

\bibitem[Arasu et~al.(2009)Arasu, R{\'e}, and Suciu]{arasu2009large}
Arasu, A., R{\'e}, C., and Suciu, D.
\newblock Large-scale deduplication with constraints using dedupalog.
\newblock In \emph{2009 IEEE 25th International Conference on Data Engineering}, pp.\  952--963. IEEE, 2009.

\bibitem[Bansal et~al.(2004)Bansal, Blum, and Chawla]{bansal2004correlation}
Bansal, N., Blum, A., and Chawla, S.
\newblock Correlation clustering.
\newblock \emph{Machine learning}, 56:\penalty0 89--113, 2004.

\bibitem[Becker(2005)]{becker2005survey}
Becker, H.
\newblock A survey of correlation clustering.
\newblock \emph{Advanced Topics in Computational Learning Theory}, pp.\  1--10, 2005.

\bibitem[Behnezhad et~al.(2019)Behnezhad, Derakhshan, Hajiaghayi, Stein, and Sudan]{behnezhad2019fully}
Behnezhad, S., Derakhshan, M., Hajiaghayi, M., Stein, C., and Sudan, M.
\newblock Fully dynamic maximal independent set with polylogarithmic update time.
\newblock In \emph{2019 IEEE 60th Annual Symposium on Foundations of Computer Science (FOCS)}, pp.\  382--405. IEEE, 2019.

\bibitem[Behnezhad et~al.(2022)Behnezhad, Charikar, Ma, and Tan]{BCMT}
Behnezhad, S., Charikar, M., Ma, W., and Tan, L.-Y.
\newblock Almost 3-approximate correlation clustering in constant rounds.
\newblock In \emph{2022 IEEE 63rd Annual Symposium on Foundations of Computer Science (FOCS)}, pp.\  720--731. IEEE, 2022.

\bibitem[Behnezhad et~al.(2023)Behnezhad, Charikar, Ma, and Tan]{behnezhad2023single}
Behnezhad, S., Charikar, M., Ma, W., and Tan, L.-Y.
\newblock Single-pass streaming algorithms for correlation clustering.
\newblock In \emph{Proceedings of the 2023 Annual ACM-SIAM Symposium on Discrete Algorithms (SODA)}, pp.\  819--849. SIAM, 2023.

\bibitem[Bonchi et~al.(2013)Bonchi, Gionis, and Ukkonen]{bonchi2013overlapping}
Bonchi, F., Gionis, A., and Ukkonen, A.
\newblock Overlapping correlation clustering.
\newblock \emph{Knowledge and information systems}, 35:\penalty0 1--32, 2013.

\bibitem[Cambus et~al.(2022)Cambus, Pai, and Uitto]{cambus2022parallel}
Cambus, M., Pai, S., and Uitto, J.
\newblock A parallel algorithm for $(3+\varepsilon)$-approximate correlation clustering.
\newblock \emph{arXiv preprint arXiv:2205.07593}, 2022.

\bibitem[Cao et~al.(2024)Cao, Cohen-Addad, Lee, Li, Newman, and Vogl]{cao2024understanding}
Cao, N., Cohen-Addad, V., Lee, E., Li, S., Newman, A., and Vogl, L.
\newblock Understanding the cluster lp for correlation clustering.
\newblock In \emph{STOC}, 2024.

\bibitem[Chakrabarty \& Makarychev(2023)Chakrabarty and Makarychev]{chakrabarty2023single}
Chakrabarty, S. and Makarychev, K.
\newblock {S}ingle-{P}ass {P}ivot {A}lgorithm for {C}orrelation {C}lustering. {K}eep it simple!
\newblock In \emph{NeurIPS 2023}, 2023.

\bibitem[Charikar et~al.(2005)Charikar, Guruswami, and Wirth]{charikar2005clustering}
Charikar, M., Guruswami, V., and Wirth, A.
\newblock Clustering with qualitative information.
\newblock \emph{Journal of Computer and System Sciences}, 71\penalty0 (3):\penalty0 360--383, 2005.

\bibitem[Charikar et~al.(2017)Charikar, Gupta, and Schwartz]{charikar2017local}
Charikar, M., Gupta, N., and Schwartz, R.
\newblock Local guarantees in graph cuts and clustering.
\newblock In \emph{International Conference on Integer Programming and Combinatorial Optimization}, pp.\  136--147. Springer, 2017.

\bibitem[Chawla et~al.(2015)Chawla, Makarychev, Schramm, and Yaroslavtsev]{chawla2015near}
Chawla, S., Makarychev, K., Schramm, T., and Yaroslavtsev, G.
\newblock Near optimal lp rounding algorithm for correlation clustering on complete and complete k-partite graphs.
\newblock In \emph{Proceedings of the forty-seventh annual ACM symposium on Theory of computing}, pp.\  219--228, 2015.

\bibitem[Chechik \& Zhang(2019)Chechik and Zhang]{chechik2019fully}
Chechik, S. and Zhang, T.
\newblock Fully dynamic maximal independent set in expected poly-log update time.
\newblock In \emph{2019 IEEE 60th Annual Symposium on Foundations of Computer Science (FOCS)}, pp.\  370--381. IEEE, 2019.

\bibitem[Cohen-Addad et~al.(2022)Cohen-Addad, Lee, and Newman]{cohen2022correlation}
Cohen-Addad, V., Lee, E., and Newman, A.
\newblock Correlation clustering with sherali-adams.
\newblock In \emph{2022 IEEE 63rd Annual Symposium on Foundations of Computer Science (FOCS)}, pp.\  651--661. IEEE, 2022.

\bibitem[Cohen-Addad et~al.(2023)Cohen-Addad, Lee, Li, and Newman]{cohen2023handling}
Cohen-Addad, V., Lee, E., Li, S., and Newman, A.
\newblock Handling correlated rounding error via preclustering: A 1.73-approximation for correlation clustering.
\newblock \emph{arXiv preprint arXiv:2309.17243}, 2023.

\bibitem[Dean \& Ghemawat(2008)Dean and Ghemawat]{dean2008mapreduce}
Dean, J. and Ghemawat, S.
\newblock Mapreduce: simplified data processing on large clusters.
\newblock \emph{Communications of the ACM}, 51\penalty0 (1):\penalty0 107--113, 2008.

\bibitem[Demaine et~al.(2006)Demaine, Emanuel, Fiat, and Immorlica]{demaine2006correlation}
Demaine, E.~D., Emanuel, D., Fiat, A., and Immorlica, N.
\newblock Correlation clustering in general weighted graphs.
\newblock \emph{Theoretical Computer Science}, 361\penalty0 (2-3):\penalty0 172--187, 2006.

\bibitem[Ding et~al.(2019)Ding, Han, Wang, Li, and Song]{ding2019user}
Ding, L., Han, B., Wang, S., Li, X., and Song, B.
\newblock User-centered recommendation using us-elm based on dynamic graph model in e-commerce.
\newblock \emph{International Journal of Machine Learning and Cybernetics}, 10:\penalty0 693--703, 2019.

\bibitem[Fang et~al.(2020)Fang, Wang, Zhao, Yu, and Wang]{fang2020dynamic}
Fang, Y., Wang, H., Zhao, L., Yu, F., and Wang, C.
\newblock Dynamic knowledge graph based fake-review detection.
\newblock \emph{Applied Intelligence}, 50:\penalty0 4281--4295, 2020.

\bibitem[Firman et~al.(2013)Firman, Thomas, Julier, and Sugimoto]{firman2013learning}
Firman, M., Thomas, D., Julier, S., and Sugimoto, A.
\newblock Learning to discover objects in rgb-d images using correlation clustering.
\newblock In \emph{2013 IEEE/RSJ International Conference on Intelligent Robots and Systems}, pp.\  1107--1112. IEEE, 2013.

\bibitem[Goodrich et~al.(2011)Goodrich, Sitchinava, and Zhang]{goodrich2011sorting}
Goodrich, M.~T., Sitchinava, N., and Zhang, Q.
\newblock Sorting, searching, and simulation in the mapreduce framework.
\newblock In \emph{International Symposium on Algorithms and Computation}, pp.\  374--383. Springer, 2011.

\bibitem[Hafiene et~al.(2020)Hafiene, Karoui, and Romdhane]{hafiene2020influential}
Hafiene, N., Karoui, W., and Romdhane, L.~B.
\newblock Influential nodes detection in dynamic social networks: A survey.
\newblock \emph{Expert Systems with Applications}, 159:\penalty0 113642, 2020.

\bibitem[Jafarov et~al.(2021)Jafarov, Kalhan, Makarychev, and Makarychev]{jafarov2021local}
Jafarov, J., Kalhan, S., Makarychev, K., and Makarychev, Y.
\newblock Local correlation clustering with asymmetric classification errors.
\newblock In \emph{International Conference on Machine Learning}, pp.\  4677--4686. PMLR, 2021.

\bibitem[J{\'a}J{\'a}(1992)]{jaja1992parallel}
J{\'a}J{\'a}, J.
\newblock \emph{Parallel algorithms}.
\newblock 1992.

\bibitem[Kalashnikov et~al.(2008)Kalashnikov, Chen, Mehrotra, and Nuray-Turan]{kalashnikov2008web}
Kalashnikov, D.~V., Chen, Z., Mehrotra, S., and Nuray-Turan, R.
\newblock Web people search via connection analysis.
\newblock \emph{IEEE Transactions on Knowledge and Data Engineering}, 20\penalty0 (11):\penalty0 1550--1565, 2008.

\bibitem[Kalhan et~al.(2019)Kalhan, Makarychev, and Zhou]{kalhan2019correlation}
Kalhan, S., Makarychev, K., and Zhou, T.
\newblock Correlation clustering with local objectives.
\newblock \emph{Advances in Neural Information Processing Systems}, 32, 2019.

\bibitem[Karloff et~al.(2010)Karloff, Suri, and Vassilvitskii]{karloff2010model}
Karloff, H., Suri, S., and Vassilvitskii, S.
\newblock A model of computation for mapreduce.
\newblock In \emph{Proceedings of the twenty-first annual ACM-SIAM symposium on Discrete Algorithms}, pp.\  938--948. SIAM, 2010.

\bibitem[Lenzen \& Wattenhofer(2010)Lenzen and Wattenhofer]{lenzen2010brief}
Lenzen, C. and Wattenhofer, R.
\newblock Brief announcement: Exponential speed-up of local algorithms using non-local communication.
\newblock In \emph{Proceedings of the 29th ACM SIGACT-SIGOPS symposium on Principles of distributed computing}, pp.\  295--296, 2010.

\bibitem[Li et~al.(2017)Li, Dau, Puleo, and Milenkovic]{li2017motif}
Li, P., Dau, H., Puleo, G., and Milenkovic, O.
\newblock Motif clustering and overlapping clustering for social network analysis.
\newblock In \emph{IEEE INFOCOM 2017-IEEE Conference on Computer Communications}, pp.\  1--9. IEEE, 2017.

\bibitem[Puleo \& Milenkovic(2016)Puleo and Milenkovic]{puleo2016correlation}
Puleo, G. and Milenkovic, O.
\newblock Correlation clustering and biclustering with locally bounded errors.
\newblock In \emph{International Conference on Machine Learning}, pp.\  869--877. PMLR, 2016.

\bibitem[Rubinfeld et~al.(2011)Rubinfeld, Tamir, Vardi, and Xie]{rubinfeld2011fast}
Rubinfeld, R., Tamir, G., Vardi, S., and Xie, N.
\newblock Fast local computation algorithms.
\newblock \emph{arXiv preprint arXiv:1104.1377}, 2011.

\bibitem[Shi et~al.(2021)Shi, Dhulipala, Eisenstat, {\L}{\u{a}}cki, and Mirrokni]{shi2021scalable}
Shi, J., Dhulipala, L., Eisenstat, D., {\L}{\u{a}}cki, J., and Mirrokni, V.
\newblock Scalable community detection via parallel correlation clustering.
\newblock \emph{Proceedings of the VLDB Endowment}, 14\penalty0 (11):\penalty0 2305--2313, 2021.

\bibitem[Swamy(2004)]{swamy2004correlation}
Swamy, C.
\newblock Correlation clustering: maximizing agreements via semidefinite programming.
\newblock In \emph{SODA}, volume~4, pp.\  526--527. Citeseer, 2004.

\bibitem[Tantipathananandh \& Berger-Wolf(2011)Tantipathananandh and Berger-Wolf]{tantipathananandh2011finding}
Tantipathananandh, C. and Berger-Wolf, T.~Y.
\newblock Finding communities in dynamic social networks.
\newblock In \emph{2011 IEEE 11th international conference on data mining}, pp.\  1236--1241. IEEE, 2011.

\bibitem[Veldt et~al.(2020)Veldt, Wirth, and Gleich]{veldt2020parameterized}
Veldt, N., Wirth, A., and Gleich, D.~F.
\newblock Parameterized correlation clustering in hypergraphs and bipartite graphs.
\newblock In \emph{Proceedings of the 26th ACM SIGKDD International Conference on Knowledge Discovery \& Data Mining}, pp.\  1868--1876, 2020.

\bibitem[Wang et~al.(2013)Wang, Xu, Chen, and Wang]{wang2013scalable}
Wang, Y., Xu, L., Chen, Y., and Wang, H.
\newblock A scalable approach for general correlation clustering.
\newblock In \emph{Advanced Data Mining and Applications: 9th International Conference, ADMA 2013, Hangzhou, China, December 14-16, 2013, Proceedings, Part II 9}, pp.\  13--24. Springer, 2013.

\bibitem[Yan et~al.(2021)Yan, Liu, Ban, Jing, and Tong]{yan2021dynamic}
Yan, Y., Liu, L., Ban, Y., Jing, B., and Tong, H.
\newblock Dynamic knowledge graph alignment.
\newblock In \emph{Proceedings of the AAAI conference on artificial intelligence}, volume~35, pp.\  4564--4572, 2021.

\bibitem[Yoshida et~al.(2009)Yoshida, Yamamoto, and Ito]{yoshida2009improved}
Yoshida, Y., Yamamoto, M., and Ito, H.
\newblock An improved constant-time approximation algorithm for maximum\~{} matchings.
\newblock In \emph{Proceedings of the forty-first annual ACM symposium on Theory of computing}, pp.\  225--234, 2009.

\end{thebibliography}
\bibliographystyle{icml2024}

\clearpage

\appendix

\newpage

\section{Illustrations for Query, Extended Query, Expensive, and Dangerous Paths}

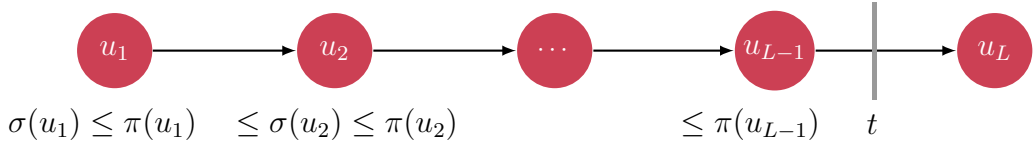
\begin{figure}[H]
    \centering
    \begin{tikzpicture}[scale=0.65,level distance=45mm,minimum size=10mm,
  sibling distance=30mm,
  every node/.style={fill=BrickRed,text=white, circle,inner sep=2pt},
  arrow/.style={thick,edge from parent/.style={draw,-latex}}
]
\node {\large$u_1$} child[grow=right,arrow] 
{
node{\large$u_2$}child[grow=right,arrow] 
{
node{$\cdots$}child[grow=right,arrow] 
{
node{\large$u_{L-1}$}child[grow=right,arrow]{
node{\large$u_{L}$}}}}};
\node[fill=none,color=black,rectangle] at (-0.25,-1.5) {\large$\sigma(u_1)\leq \pi(u_{1})$};
\node[fill=none,color=black,rectangle] at (4.75,-1.5) {\large$\leq \sigma(u_2)\leq \pi(u_{2})$};
\node[fill=none,color=black,rectangle] at (15.5,-1.5) {\large$t$};
\node[fill=none,color=black,rectangle] at (13,-1.5) {\large$\leq\pi(u_{L-1})$};
\fill[black!40!white] (15.5,-1) rectangle (15.6,1);
\end{tikzpicture}
    \caption{Path $(u_1,\dots,u_{L-1},u_L)$ is a query path if each $u_i$ ($i>1$) queries $u_{i-1}$. This condition is equivalent to 
    $\sigma(u_1)\leq \pi(u_1) \leq \cdots \leq \sigma(u_L) \leq \pi(u_L)$. Note that $\sigma(u)\leq \pi(u)$ for every node $u\in V$.
    Set $\calQ_t(u_1,u_2)$ contains this path for $t\geq \pi(u_{L-1})$.}
    \label{fig:query-path}
\end{figure}

\begin{figure}[H]
    \centering
    \begin{tikzpicture}[scale=0.65,level distance=35mm,
  sibling distance=30mm,,minimum size=10mm,
  every node/.style={fill=BrickRed,text=white, circle,inner sep=2pt},
  arrow/.style={thick,edge from parent/.style={draw,-latex}}
]
\node {\large$u_1$} child[grow=right,arrow] 
{
node{\large$u_2$}child[grow=right,arrow]
{
node{$\cdots$}child[grow=right,arrow] 
{
node{$u_{L-2}$}child[grow=right,arrow]
{
node{\large$u_{L-1}$}child[grow=right,arrow]{
node[fill=black]{\large$u_{L}$}}}}}};
\end{tikzpicture}
    \caption{Path $(u_1,\dots,u_{L-1},u_L)$ is an extended query path (EQ-path) if each vertex on the path, except for the first and last one, queries the previous vertex and $\pi(u_{L-2})\leq \sigma(u_L)$. This condition is equivalent to 
    $\sigma(u_1)\leq \pi(u_1) \leq \cdots\leq \pi(u_{L-2}) \leq \min(\sigma(u_{L-1}), \sigma(u_{L}))$. The last inequality says that neither $u_{L-1}$ nor $u_{L}$ is settled before $u_{L-2}$ is processed.}    
    \label{fig:eq-path}
\end{figure}
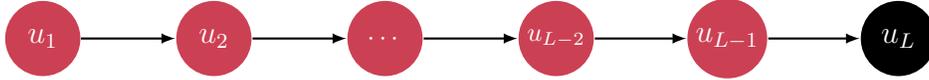

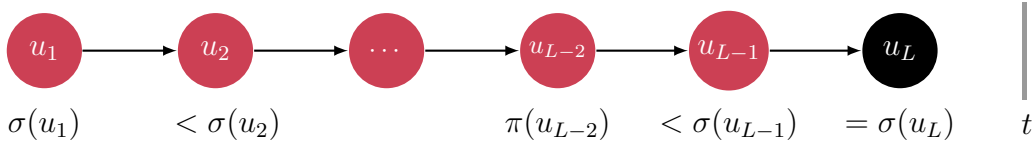
\begin{figure}[H]
    \centering
    \begin{tikzpicture}[scale=0.65,level distance=35mm,,minimum size=10mm,
  sibling distance=30mm,
  every node/.style={fill=BrickRed,text=white, circle,inner sep=2pt},
  arrow/.style={thick,edge from parent/.style={draw,-latex}}
]
\node {\large$u_1$} child[grow=right,arrow] 
{
node{\large$u_2$}child[grow=right,arrow]
{
node{$\cdots$}child[grow=right,arrow] 
{
node{$u_{L-2}$}child[grow=right,arrow]
{
node{\large$u_{L-1}$}child[grow=right,arrow]{
node[fill=black]{\large$u_{L}$}}}}}};
\node[fill=none,color=black,rectangle] at (0,-1.5) {\large$\sigma(u_1)$};
\node[fill=none,color=black,rectangle] at (3.75,-1.5) {\large$<\sigma(u_2)$};
\node[fill=none,color=black,rectangle] at (20.1,-1.5) {\large$t$};
\node[fill=none,color=black,rectangle] at (10.5,-1.5) {\large$\pi(u_{L-2})$};
\node[fill=none,color=black,rectangle] at (14,-1.5) {\large$<\sigma(u_{L-1})$};
\node[fill=none,color=black,rectangle] at (17.5,-1.5) {\large$=\sigma(u_{L})$};
\fill[black!40!white] (20,-1) rectangle (20.1,1);
\end{tikzpicture}
    \caption{Path $(u_1,\dots,u_{L-1},u_L)$ is an expensive extended query path (expensive EQ-path) if 
    (1) each vertex on the path, except for the first and last one, queries the previous vertex, (2) $\sigma(u_1)< \sigma(u_2)$, and (3) $\pi(u_{L-2})\leq \sigma(u_{L-1}) =\sigma(u_L)$. This condition is equivalent to 
    $\sigma(u_1)\leq \pi(u_1) <\sigma(u_2)\leq \cdots\leq \pi(u_{L-2}) \leq \sigma(u_{L-1}) = \sigma(u_L)$. Inequality $\sigma(u_1)\leq \pi(u_1) <\sigma(u_2)$ tells us that \pivot places $u_1$ and $u_2$ in distinct clusters (i.e., edge $(u_1,u_2)$ is cut by \pivot). Condition $\sigma(u_{L-1}) = \sigma(u_L)$ tells us that \pivot places $u_{L-1}$ and $u_{L}$ in the same cluster. Set $\calX_t(u_1,u_2)$ contains this path for $t\geq \sigma(u_{L-1})=\sigma(u_L)$.}   
    \label{fig:expensive-eq-path}
\end{figure}

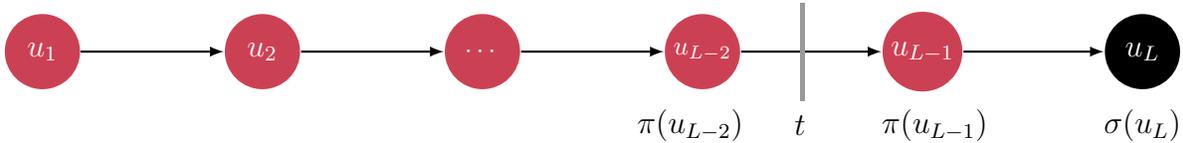
\begin{figure}[H]
    \centering
    \begin{tikzpicture}[scale=0.65,level distance=45mm,minimum size=10mm,
  every node/.style={fill=BrickRed,text=white, circle,inner sep=2pt},
  arrow/.style={thick,edge from parent/.style={draw,-latex}}
]
\node {\large$u_1$} child[grow=right,arrow] 
{
node{\large$u_2$}child[grow=right,arrow]
{
node{$\cdots$}child[grow=right,arrow] 
{
node{$u_{L-2}$}child[grow=right,arrow]
{
node{\large$u_{L-1}$}child[grow=right,arrow]{
node[fill=black]{\large$u_{L}$}}}}}};
\fill[black!40!white] (15.5,-1) rectangle (15.6,1);
\node[fill=none,color=black,rectangle] at (15.5,-1.5) {\large$t$};
\node[fill=none,color=black,rectangle] at (13.25,-1.5) {\large$\pi(u_{L-2})$};
\node[fill=none,color=black,rectangle] at (18.25,-1.5) {\large$\pi(u_{L-1})$};
\node[fill=none,color=black,rectangle] at (22.5,-1.5) {\large$\sigma(u_{L})$};
\end{tikzpicture}
    \caption{Path $(u_1,\dots,u_{L-1},u_L)$ is dangerous at time $t$ if it is an extended query path and 
    $\pi(u_{L-2})\leq t$ but $\pi(u_{L-1})>t$ and $\sigma(u_L)> t$.
}   
    \label{fig:dangerous}
\end{figure}

\newpage

\section{Empirical Evaluation}
\label{sec:empirical}

In this section, we conduct a simple empirical assessment of our algorithm. 
We compare the correlation clustering cost of \ourPivot to that of \Pivot, R-\pivot of \citet{BCMT}, and Narrow-\pivot of \citet{chakrabarty2023single} on synthetic graphs.  

\paragraph{Set-Up.} We use stochastic block model graphs generated as follows: Each sample graph has three partitions, each with 200 nodes. The probability of the appearance of an edge inside each partition is $0.9$, and between partitions is $0.1$. 

We generate 100 graphs and run \pivot, $R$-\pivot, Narrow-\Pivot, and \ourPivot on these graphs with parameter $R$ ranging from $2$ to $30$ (Here $R$ is the parameter $k$ for \ourPivot). Note that \pivot does not depend on $R$. For each $R$, we take the mean error of these runs for each algorithm (see \cref{fig:experiments}). We remove the standard deviation of the error for figure readability; the std is very similar for all the algorithms for $R>13$ and is around $2300$.

\begin{figure}
    \centering
    \includegraphics[width=0.5\columnwidth]{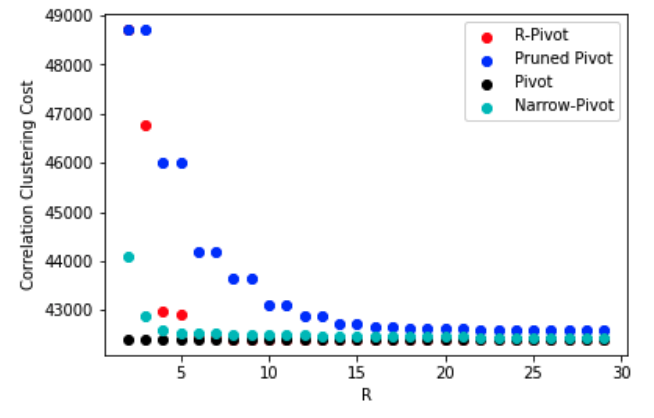}
    \caption{Comparison of correlation clustering cost for \pivot, $R$-\pivot, Narrow-\pivot and \ourPivot. The optimal clustering has an expected cost of less than 17970.\label{fig:experiments}}
\end{figure}

\paragraph{Results.} First note that the clustering that puts each partition in one cluster achieves an expected cost of $17970$, which is an upper bound on the average optimum value. 
We selected edge probabilities $0.1$ and $0.9$ partly to approximate the optimal clustering cost, as computing the exact value is NP-hard.
So, \Pivot's approximation factor is at least $2.35$ in this case. Furthermore, the trivial clustering that puts each node in one cluster results in $65730$ average error which is significantly more than the error of the other algorithms. 

We observe that all three \pivot variants converge fast to the cost of \pivot. 
Even though, as a function of $R$, \ourPivot has the steadiest improvement in cost, it still converges to the cost of \pivot exponentially.
\ourPivot queries significantly fewer nodes compared to the other algorithms, albeit at the expense of a negligible increase in its approximation factor. In fact, for $R\ge 4$, the increase in the cost of \ourPivot compared to \Pivot is less than $\%1$.  
This property of \ourPivot makes it flexible in adapting to parallel and dynamic settings while losing a small factor in the approximation guarantee. 
\section{Equivalence between \Pivot and \text{Sequential \pivot}}
In \cref{alg:standard-pivot} we recall the \pivot algorithm from \cite{pivot}.
\begin{algorithm}[h]
\caption{\textsc{\Pivot} \label{alg:standard-pivot}}
\begin{algorithmic}[1]
    \STATE \textbf{function} $\pivot(G = (V, E))$
    \STATE \ourIndent[1] \textbf{if} $V = \emptyset$:
    \STATE \ourIndent[2] terminate
    \STATE \ourIndent[1] Pick a random pivot $u \in V$.
    \STATE \ourIndent[1] Cluster together $u$ and its neighbors.
    \STATE \ourIndent[1] Let $H$ be obtained by removing $u$ and its neighbors from $G$.
    \STATE \ourIndent[1] $\pivot(H)$
    \end{algorithmic}
\end{algorithm}

\begin{lemma}\label{lem:pivot_equival}
    \cref{alg:standard-pivot} and \cref{alg:sequential-pivot} are equivalent
\end{lemma}
\begin{proof}
    Instead of choosing a random pivot $u$ each time on Line~4 of \pivot (\cref{alg:standard-pivot}), we assume that before the algorithm is invoked a random permutation $\pi$ over the input nodes is chosen.
    Then, Line~4 is implemented by choosing the first, with respect to $\pi$, node available in $V$ in that invocation of \pivot.

    We prove the lemma by induction on $\pi$. First, consider the node with rank $1$, i.e., let $\pi(u)= 1$. $u$ is a pivot in \pivot, and since $u$ has no higher ranked neighbors, it executes line 8 in \textsc{Sequential \Pivot} and hence is a pivot there as well.

    Suppose that for some $t\ge 1$, both \pivot and \textsc{Sequential \Pivot} cluster all nodes with ranks $1,\ldots, t$ the same way. Suppose $\pi(u)=t+1$. First, suppose that $u$ is a pivot in \pivot. Then it must be that $u$ does not have any pivot neighbor with a higher rank: if there is such node $v$, then when $v$ is being clustered, $u$ is put in the cluster of $v$ in Line~5 of \cref{alg:standard-pivot}. This means that in \textsc{Sequential \Pivot} $u$ is also a pivot. 
    
    Now, suppose that $u$ is not a pivot in \pivot. Let $v$ be the pivot of $u$ in \pivot. Then it must be that $v$ is the highest ranked pivot in the neighborhood of $u$: if there is a higher ranked pivot $v'$ in the neighborhood of $u$, then when $v'$ is being processed (before $v$), $u$ is put in the cluster of $v'$ in Line~5 of \cref{alg:standard-pivot}. This means that in \textsc{Sequential \Pivot}, when $u$ is being processed, Line~7 is executed when the neighbor $v$ of $u$ is picked. 
\end{proof}
\end{document}